%% file: paper_complexity.tex
\documentclass{article}
\usepackage{arxiv}
\usepackage[utf8]{inputenc}
\usepackage[T1]{fontenc}
\usepackage[english]{babel}
\usepackage{amsmath}
\usepackage{amsthm}
\usepackage{amssymb,amsfonts}
\usepackage{mathrsfs}
\usepackage{bbm}
\usepackage{graphicx}
\usepackage{enumitem}
\usepackage{xcolor}
\RequirePackage[colorlinks,citecolor=blue,urlcolor=blue]{hyperref}
\usepackage[normalem]{ulem}
\usepackage{pgf,tikz}

\usetikzlibrary{arrows,automata,shapes}
\usetikzlibrary{positioning}
\usepackage{dsfont}
\usepackage{fancyhdr}
\usepackage{mathtools}
\usepackage[ruled,vlined]{algorithm2e}
\usepackage{caption}
\usepackage{graphicx}
\usepackage{subcaption}
\usepackage{booktabs} 
\usepackage{parcolumns}
\usepackage{breqn}

\newcommand{\T}{\mathcal{T}}

\usepackage{arydshln}
\newtheorem{axiom}{Axiom}[section]
\newtheorem{definition}[axiom]{Definition}
\newtheorem{thm}{Theorem}[section]
\newtheorem{lem}[thm]{Lemma}
\newtheorem{rem}[thm]{Remark}
\newtheorem{prop}[thm]{Property}
\newtheorem{cor}[thm]{Corollary}

\def\independenT#1#2{\mathrel{\rlap{$#1#2$}\mkern2mu{#1#2}}}
\newcommand{\stkout}[1]{\ifmmode\text{\sout{\ensuremath{#1}}}\else\sout{#1}\fi}

\makeatletter
\renewcommand*\env@matrix[1][*\c@MaxMatrixCols c]{%
  \hskip -\arraycolsep
  \let\@ifnextchar\new@ifnextchar
  \array{#1}}
\makeatother

\newcommand{\boxxx}[1]
 {\begin{center}\fbox{\begin{minipage}{12.50cm}#1\smallskip\end{minipage}}\end{center}}

\newcommand{\MCN}{MCN}
\newcommand{\cand}{candidate}
\newcommand{\eg}{\emph{e.g.}}
\newcommand{\ie}{\emph{i.e.}}
\newcommand{\etal}{\emph{et al.}}

\definecolor{brass}{rgb}{0.71, 0.65, 0.26}

\newcommand{\teal}[1]{\textcolor{teal}{#1}}
\newcommand{\margarida}[1]{\teal{\tiny \textsc{MC:} #1}}

\newcommand{\magenta}[1]{\textcolor{magenta}{#1}}

\title{Complexity of the Multilevel Critical Node Problem}

\author{
  Adel Nabli \qquad  Margarida Carvalho
  \\
  CIRRELT and D\'epartement d'Informatique et de Recherche Op\'erationnelle\\
   Universit\'e de Montr\'eal\\
  \texttt{adel.nabli@umontreal.ca}\\
  \texttt{carvalho@iro.umontreal.ca} \\
  \And
Pierre Hosteins \\
 COSYS-ESTAS\\
 Universit\'e Gustave Eiffel\\
 \texttt{pierre.hosteins@ifsttar.fr} 
}

\begin{document}

\maketitle

\begin{abstract}
   In this work, we analyze a sequential game played in a graph called the Multilevel Critical Node problem (\MCN). A defender and an attacker are the players of this game. The defender starts by preventively interdicting vertices (vaccination) from being attacked. Then, the attacker infects a subset of non-vaccinated vertices and, finally, the defender reacts with a protection strategy. We provide the first computational complexity results associated with {\MCN} and its subgames. Moreover, by considering unitary, weighted, undirected, and directed graphs, we clarify how the theoretical tractability of those problems vary. Our findings contribute with new NP-complete, $\Sigma_2^p$-complete and $\Sigma_3^p$-complete problems. Furthermore, for the last level of the game, the protection stage, we build polynomial time algorithms for certain graph classes.
\end{abstract}


\input{Paper_Version/Introduction}

\input{Paper_Version/Undirected_Graphs_unitary}

\input{Paper_Version/Undirected_Graphs_weighted_adel_margarida}

\input{Paper_Version/Directed_Graphs}

\input{Paper_Version/Special_Cases}

\section*{Acknowledgements}

The authors  wish to thank the support of the Institut de valorisation des donn\'ees and Fonds de Recherche du Qu\'ebec through the FRQ–IVADO Research Chair in Data Science for Combinatorial Game Theory, and the Natural Sciences and Engineering Research Council of Canada through the discovery grant 2019-04557.

\bibliographystyle{plain}
\bibliography{Biblio}

\end{document}

%% file: Paper_Version/Introduction.tex
\section{Introduction}

\paragraph{Multilevel Critical Node}
 Graphs are powerful mathematical structures that enable us to model real-world networks. The problem of breaking the connectivity of a graph has been extensively studied in combinatorial optimization since it can serve to measure the robustness of a network to disruptions.  In this work, we will focus on the Multilevel Critical Node problem (\MCN)~\cite{Baggio2020}. Let $G=(V,A)$ be graph with a set $V$ of vertices and a set $A$ of arcs. 
 In {\MCN} there are two players, designated by defender and attacker, whose individual strategies are given by a selection of subsets of $V$. The game goes as follows: first, the defender selects a subset of vertices $D \subseteq V$ to \emph{vaccinate} subject to a budget limit $\Omega$ and a cost$\{\hat{c}_v\}_{v \in V}$; second, the attacker observes the vaccination strategy, and selects a subset of vertices $I \subseteq V\setminus D$ to (directly) \emph{infect} subject to a budget limit $\Phi$ and a cost$\{h_v\}_{v \in V}$; and third, the defender observes the infection strategy, and selects a subset of vertices $P \subseteq V\setminus I$ to \emph{protect} subject to a budget limit $\Lambda$ and a cost$\{c_v\}_{v \in V}$. An infected vertex $v$ propagates the infection to a vertex $u$, if $(v,u) \in A$ and $u$ is neither a vaccinated or a protected vertex. The goal of the defender is to maximize the benefit $b_v$ of saved vertices (\ie, not infected), while the attacker aims to minimize it. We assume that all parameters of the problem are non-negative integers. The game description can be succinctly given by the following mixed integer trilevel program:
\begin{subequations}%
\begin{alignat}{4}
(MCN) \ \ 
\max_{ \begin{subarray}{c}\\
 z \in \{0,1\}^{|V|} \\[1pt]
\displaystyle \sum_{v\in V} \hat{c}_v z_{v} \leq \Omega
 \end{subarray}} \quad
 & 
 \min_{  \begin{subarray}{c}\\ 
 y \in \{0,1\}^{|V|}\\[1pt] 
\displaystyle \sum_{v\in V} h_v y_{v} \leq \Phi
 \end{subarray}} \quad
 &
 \max_{ 
 \begin{subarray}{c}\\ 
  x \in \{0,1\}^{|V|} \\[1pt] 
 \alpha \in \left[0,1\right]^{|V|}
 \end{subarray}
 } \ \ &\sum_{v\in V} b_v \alpha_{v} \nonumber\\
&  &  s.t.&   \sum_{v\in V} c_v x_{v}  \leq   \Lambda  \nonumber& \\
& &  &\alpha_{v}  \leq 1 + z_{v} - y_{v} & \qquad \forall v\in V  \label{const:inf_vac} \\
& &  &\alpha_{v}  \leq  \alpha_{u} + x_{v} + z_{v} & \ \ \forall\ (u,v)\in A, \label{const:proc_vac}
\end{alignat}
\label{MCN_formulation}%
\end{subequations}
where $z$, $y$, $x$ and $\alpha$ are decision vectors which coordinates are $z_v$, $y_v$, $x_v$ and $\alpha_v$ for each $v \in V$. In this optimization model, $z$, $y$ and $x$ reflect the set of vaccinated vertices $D=\lbrace v\in V: z_v=1 \rbrace$, directly infected vertices $I=\lbrace v\in V: y_v=1 \rbrace$ and protected vertices $P=\lbrace v\in V: x_v=1 \rbrace$, respectively. Finally, $\alpha$ mimics the propagation of the infection among the vertices in $V$, through Constraints~\eqref{const:inf_vac} and~\eqref{const:proc_vac}, and it is necessarily binary due to the maximization in the last level (protection). Concretely, $\alpha_v=1$ means that vertex $v$ is saved and $\alpha_v=0$ means that vertex $v$ is infected. In multilevel optimization, the first stage (in {\MCN}, the vaccination stage) is called the upper level or first level, the second stage is called the second level, and so on, with the last stage being also designated by lower level. See~\cite{Baggio2020} for further details on this mathematical programming formulation and Figure~\ref{fig:example} for an illustration of the game. 

\input{Paper_Version/Example.tex}

\paragraph{Contributions} To the best of our knowledge, this work is the first providing a computational complexity classification of the decision version of {\MCN}, as well as, of its subgames. Namely, we investigate the subgames \emph{(i)} {\sc Protect}, where given $D$ and $I$, the defender seeks the optimal protection strategy, \emph{(ii)} {\sc Attack}, where given $D$ and no protection budget, the attacker determines the optimal infection strategy, \emph{(iii)} {\sc Attack-Protect}, where given $D$, the attacker computes the optimal infection strategy, and \emph{(iv)} {\sc Vaccination-Attack}, where given no budget for protection, the defender finds the optimal vaccination strategy. This fundamental contribution sheds light on the practical difficulties dealt in~\cite{Baggio2020}. Furthermore, it contributes to the understanding of sequential combinatorial games within the polynomial hierarchy and it motivates the focus on potentially $\Omega(2^{2^{2^{\vert V \vert}}})$ algorithms, heuristic methods or novel solution definitions. Table~\ref{tab:our_results} summarizes our results for general graphs; unitary cases assume that all costs and benefits are 1, and undirected graphs assume that infection can traverse an edge in both directions. We stress the incorrectness of the following intuitive claim for multilevel optimization problems: if a subgame is $C$-hard for some complexity class $C$, then the associated game is at least $C$-hard. Note that in a multilevel optimization problem, like the {\MCN}, the ultimate goal is to find the optimal first level decision. Hence, if for example in the {\MCN}, we had always $\Omega =\vert V \vert$, then we would know directly that all vertices are saved, even if the attack problem is theoretically intractable. This supports the interest of understanding the individual complexity of each subgame of {\MCN}.

We also contribute with an algorithmic analyzes of \textsc{Protect} by exploring graph classes where it becomes polynomially solvable.

\begin{table}[h]
\begin{center}
\resizebox{0.95\columnwidth}{!}{
\begin{tabular}{c|ll|ll}
\toprule
\multicolumn{1}{c|}{} & \multicolumn{2}{c|}{\textbf{Undirected Graphs}} & \multicolumn{2}{c}{\textbf{Directed Graphs}} \\[1ex]
\multicolumn{1}{c|}{\textsc{Decision Versions}} & \multicolumn{1}{c}{\textsc{Unitary Case}} & \multicolumn{1}{c|}{\textsc{Weighted Case}} & \multicolumn{1}{c}{\textsc{Unitary Case}} & \multicolumn{1}{c}{\textsc{Weighted Case}} \\
\multicolumn{1}{c|}{} & \multicolumn{1}{c}{\textit{Section \ref{sec:undirected_uni}}} & \multicolumn{1}{c|}{\textit{Section \ref{sec:undirected_wei}}} & \multicolumn{1}{c}{\textit{Sections \ref{sec:directed_uni} \& \ref{sec:special}}} & \multicolumn{1}{c}{} \\
\midrule 
\textsc{Protect} & \hyperref[th_protect_NP]{[1]} NP-complete& [6] \textcolor{gray}{NP-complete} & \hyperref[th_protect_dir]{[11]} NP-complete & [16] \textcolor{gray}{NP-complete}\\[0.5ex]
\textsc{Attack} & \hyperref[th_VA_NP]{[2]} Polynomial& \hyperref[Attack_w_KP]{[7]} NP-complete & \hyperref[A_dir_NP_hard]{[12]} NP-complete & [17] \textcolor{gray}{NP-complete}\\[0.5ex]
\textsc{Attack-Protect} & \hyperref[th_AP_NP]{[3]} NP-hard& \hyperref[th_AP_w]{[8]} $\Sigma_2^p$-complete & [13] \textcolor{gray}{NP-hard}& [18] \textcolor{gray}{$\Sigma_2^p$-complete}\\[0.5ex]
\textsc{Vaccination-Attack} & \hyperref[cor_VA_NP]{[4]} NP-complete& \hyperref[th_VA_w]{[9]} $\Sigma_2^p$-complete& \hyperref[th_VA_dir]{[14]} $\Sigma_2^p$-complete&[19] \textcolor{gray}{$\Sigma_2^p$-complete}\\[0.5ex]
\textsc{MCN} & [5] \textcolor{gray}{NP-hard} & \hyperref[th_MCN_w]{[10]} $\Sigma_3^p$-complete &[15] \textcolor{gray}{$\Sigma_2^p$-hard}& [20] \textcolor{gray}{$\Sigma_3^p$-complete}\\[0.5ex]
\bottomrule
\end{tabular}}
\vskip 0.10in
\caption{\normalsize{Computational complexity of the decision versions of the subproblems in MCN. Entries in \textcolor{gray}{gray} correspond to results that follow as corollaries. In increasing order, we have: $[4] \implies [5]$, $[1] \implies [6]$,  $[12] \implies [13]$, $[14] \implies [15]$, and $[6$-$10] \implies [16$-$20]$.}}
\label{tab:our_results}
\end{center}
\end{table}

\paragraph{Paper Organization}  In Section~\ref{sec:literature}, we revise the literature associated with {\MCN}, allowing to position our contribution in the context of critical node problems, interdiction games and defender-attacker-defender problems. In Section~\ref{sec:undirected_uni}, we focus on the case where graphs are undirected and each vertex benefit and cost is unitary. Section~\ref{sec:undirected_wei} adds  the possibility of having non-unitary parameters, while Section~\ref{sec:directed_uni} 
generalizes the game to directed graphs. Finally, Section~\ref{sec:special} investigates structural properties of special graph classes that can be explored to make at least {\sc Protection} polynomially solvable, both on directed and undirected graphs.

\section{Related literature} \label{sec:literature}
 Assessing the vulnerability of complex infrastructures such as networks is of the utmost importance in practice. One way to measure the robustness of a given network is to study its connectivity properties, for which many metrics exist. With respect to a fixed metric, vertices often play different roles in the graph, with varying levels of importance. The most important vertices are  qualified as \emph{critical}. Thus, the problem of detecting subsets of critical vertices with respect to some connectivity measure is of great interest, either for defensive or for offensive purposes, and with applications in domains ranging from network immunization \cite{ARULSELVAN20092193, He2011} to computational biology \cite{Boginski2009, Tomaino2012}.

 \paragraph{Critical Node Detection Problems (CNDP)} The CNDPs have been extensively studied, with names varying with the connectivity metric to optimize and the constraints of the problem. Many of its studied versions have been shown to be NP-complete on general graphs; see Lalou \etal ~\cite{LALOU201892} for a recent survey. Indeed, many of these belong to the class of problems called \textit{Node-Deletion Problems}. They consist in deleting the smallest subset of vertices from a graph so that the induced subgraph satisfies a certain property $\pi$. Lewis and Yannakakis \cite{Lewis1980} have shown that if $\pi$ is \emph{nontrivial} and \emph{hereditary}, then the subsequent vertex deletion problem is NP-hard. In particular, \emph{MinMaxC}, the problem of finding a set of vertices $D$ from a graph $G$ with a budget constraint $|D| \leq \Omega$ such that the removal of $D$ minimizes the size of the largest connected component in the remaining graph, has been shown to be NP-hard in the strong sense thanks to this argument \cite{Shen2012}. Moreover, some CNDP problems remain NP-hard even on particular graph classes \cite{Addis2013, LALOU201892}. For example,  the original \textit{Critical Node Problem} (CNP) \cite{ARULSELVAN20092193} which seeks to minimize the \emph{pairwise connectivity} of the graph by removing a limited number of vertices remains NP-hard on split or bipartite graphs \cite{Addis2013}. Several works try to clarify the frontier between polynomial and NP-hard instances for different variants of the CNDP. The version based on pairwise connectivity over trees is studied in \cite{DisGroLoc11cnp} where it is found to be polynomial with unit connection costs and strongly NP-hard otherwise. Many other versions of the CNDP were studied in details over trees, such as the versions based on the cardinality of the largest component (\emph{MinMaxC}) and the number of connected components (\emph{MaxNum}) \cite{Shen2012}, the largest pairwise connectivity among all components \cite{Lalou2016} or an extension of pairwise connectivity based on the length of shortest paths in the remaining graph \cite{DAM-DCNP}. A stochastic version of the pairwise CNDP with node attack failure was studied over trees in \cite{stoch_tree} and found to be strongly NP-hard, even with unit connection costs. The CNDP has also been studied on other specially structured graphs, such as series-parallel graphs \cite{DAM-DCNP,Shen2012}, graphs with bounded treewidth \cite{Addis2013}, proper interval graphs \cite{Lalou2016} or bipartite permutation graphs \cite{Lalou2019}.

\paragraph{Interdiction Games} In several CNDP, although the optimization problem is formulated with a natural single objective, the task is inherently constituted of several ones. In the CNP, minimizing the pairwise connectivity maximizes the
number of connected components in the residual graph, while simultaneously minimizing the variance in the component sizes \cite{ARULSELVAN20092193}. Even though in this particular case, it has been shown that the multi-objective formulation is not equivalent to the original one \cite{Ventresca2018}, splitting the objective in two is sometimes possible. For example, Furini \etal ~\cite{Furini2019CastingLO} exhibited the hidden bilevel structure of the \textit{Capacitated Vertex Separator problem} by formulating it as  a two player Stackelberg game in which a leader interdicts the network by removing some of its vertices and a follower determines the maximum connected component in the remaining graph, highlighting the link between CNDP problems and \textit{Interdiction Games}. Interdiction games on networks are a special family of two-player zero-sum Stackelberg games in which a leader interdicts parts of the network \textit{(arcs or vertices)} subject to a budget limitation in order to maximize the disruption of the follower's objective who solves an optimization problem on the remaining graph \textit{(\eg, the maximum flow or the maximum clique)}. Whereas some interdiction games such as the \textit{network flow interdiction} are NP-complete \cite{Wood1993}, others such as the \textit{binary knapsack interdiction problem} \cite{DeNegre2011, Caprara2014} or the \textit{maximum clique interdiction game} \cite{Furini2019} have been shown to be $\Sigma_2^p$-complete, shading light on the intrinsic relationship between this class of problems and the second level of the polynomial hierarchy.

However, the unitary undirected version of {\MCN}, as originally introduced by Baggio \etal~\cite{Baggio2020}, is not an interdiction problem per se but \emph{contains} one. Indeed, the vaccination stage of the game focuses on identifying critical infrastructures in the network to interdict them preventively to counter an intentional attack, which falls into the framework of Network Interdiction problems. Nevertheless, the game does not finish with the attack: there is a third stage where the defender tries to isolate the propagation of the infection to maximize the unharmed fraction of the network. Finding a blocking strategy to limit the diffusion of an infection is related to the \textit{Firefighter problem}, which has been shown to be NP-complete, even for trees of maximum degree three \cite{Finbow2007} and was studied more recently in \cite{Barnetson2020} where the problem was shown to be NP-complete on split graphs and bipartite graphs but polynomial on graphs with bounded treewidth. Thus, the MCN problem combines two different paradigms in network protection, \textit{prevention} and \textit{blocking}, each being related to provably hard problems. The overall contraction leads to a trilevel optimization formulation for the MCN, making it fall under the \textit{Defender-Attacker-Defender} (DAD) framework introduced by Brown \etal ~\cite{Brown2006} to study the defense of critical infrastructure against malicious attacks.

\paragraph{Defender-Attacker-Defender}
Although the general DAD has been claimed to be NP-hard in \cite{Martin07}, complexity results for trilevel combinatorial optimization problems are scarce. In \cite{Johannes2011NewCO}, a new proof that \textit{Trilevel Linear Programming} is $\Sigma_2^p$-hard is provided, building upon the results in  \cite{Blair1992, Dudas1998, Jeroslow1985} showing that the \textit{Multilevel Linear Programming} problem with $L+1$ levels is $\Sigma_L^p$-hard. In fact, the decision version of MCN problem can be formulated as \textit{"given $3$ integer budgets $\Omega, \Phi, \Lambda$, a graph $G$ and an integer $K$, is there a vaccination $D$ such that for all attacks $I$ there exists a protection $P$ saving at least $K$ vertices?"} Thus, there seems to be a link between the MCN and the \textit{$3$-alternating quantified satisfiability} problem which has been shown to be $\Sigma_3^p$-complete by Meyer, Stockmeyer and Wrathall \cite{MeyerStock72, Wrathall1976}, making one expect the MCN to be complete for this class.

We stress that very few problems have been shown to be \emph{naturally} $\Sigma_3^p$-complete in the literature up to now, in addition to infinite families of problems which have been shown to be $\Sigma_L^p$-complete for any level $L$ of the polynomial hierarchy (as, \eg, Satisfiability Problems, or Multilevel Linear Programming). The compendium of \cite{Schaefer2002}, whose last update dates back to 2008, describes eight $\Sigma_3^p$-complete problems including graph theory problems, problems over sets as well as number theory problems. Since this compendium was last updated, a handful of other problems have been demonstrated to be $\Sigma_3^p$-complete, in the domains of logic, knowledge representation and artificial intelligence. We can cite, \eg, the problem of \emph{Binding Forms in First-Order Logic} \cite{Mogavero2015}, deciding whether a propositional program has epistemic FLP (Faber, Leone and Pfeifer) answer sets \cite{Shen2016}, or checking the existence of max optimal outcomes over $m$CP-nets to study the aggregation of preferences over combinatorial domains in artificial intelligence \cite{Malizia2018}. To the best of our knowledge, there is approximately a dozen proven natural $\Sigma_3^p$-complete problems in the literature, which makes it all the harder to derive $\Sigma_3^p$-hardness for a given trilevel problem. In this work, we add two more problems to the list of $\Sigma_3^p$-complete problems, the \emph{Trilevel Knapsack Interdiction Problem} and the \emph{Multilevel Critical Node Problem}. Even though the set of proven $\Sigma_2^p$-complete problems is larger by one order of magnitude, \ie, a little more than roughly a hundred such problems are known, we also add several new $\Sigma_2^p$-complete problems to this list.

%% file: Paper_Version/Example.tex
\begin{figure}\centering
\resizebox{\textwidth}{!}{
\begin{tabular}{|c|c|c|c|}\hline
&&&\\
Instance & Vaccinate 3 & Infect 2 & Protect 1\\
\begin{tikzpicture}[->,>=stealth',shorten >=0.2pt,auto,node distance=1cm,baseline={(current bounding box.north)},  scale=0.8, transform shape,
  main node/.style={circle,fill=black!15,draw,minimum size=0.7cm},spe node/.style={fill=white,draw=white}]

  \node[main node] (1) {1};
  \node[main node] (3) [below right=of 1] {3};
  \node[main node] (2) [above right =of 3] {2};
  \node[main node] (4) [below left=of 3] {4};
  \node[main node] (5) [below right=of 3] {5};
  \node[main node] (6)  [right= of 3] {6};

  \path[every node/.style={font=\sffamily\small}]
    (1) edge node {}  (4)
    (2) edge node {}  (1)
        edge node {}  (6)
    (3) edge node {}  (1)
        edge node {}  (2)
        edge node {}  (5)
    (4) edge node {}  (3)
    (5) edge node {}  (3);
\end{tikzpicture} &
\begin{tikzpicture}[->,>=stealth',shorten >=0.2pt,auto,node distance=1cm,baseline={(current bounding box.north)},  scale=0.8, transform shape,
  main node/.style={circle,fill=black!15,draw,minimum size=0.7cm},spe node/.style={fill=white,draw=white}]

  \node[main node] (1) {1};
  \node[main node] (3) [below right=of 1] {3};
  \node[main node] (2) [above right =of 3] {2};
  \node[main node] (4) [below left=of 3] {4};
  \node[main node] (5) [below right=of 3] {5};
  \node[main node] (6)  [right= of 3] {6};

  \path[every node/.style={font=\sffamily\small}]
    (1) edge node {}  (4)
    (2) edge node {}  (1)
        edge node {}  (6);
\end{tikzpicture}
&
\begin{tikzpicture}[->,>=stealth',shorten >=0.2pt,auto,node distance=1cm,baseline={(current bounding box.north)},  scale=0.8, transform shape,
  main node/.style={circle,fill=black!15,draw,minimum size=0.7cm},spe node/.style={circle,fill=black,draw,minimum size=0.7cm}]

  \node[main node] (1) {1};
  \node[main node] (3) [below right=of 1] {3};
  \node[spe node] (2) [above right =of 3] {\textcolor{white}{2}};
  \node[main node] (4) [below left=of 3] {4};
  \node[main node] (5) [below right=of 3] {5};
  \node[main node] (6)  [right= of 3] {6};

  \path[every node/.style={font=\sffamily\small}]
    (1) edge node {}  (4)
    (2) edge node {}  (1)
        edge node {}  (6);
\end{tikzpicture}
&
\begin{tikzpicture}[->,>=stealth',shorten >=0.2pt,auto,node distance=1cm,baseline={(current bounding box.north)},  scale=0.8, transform shape,
  main node/.style={circle,fill=black!15,draw,minimum size=0.7cm},spe node/.style={circle,fill=black,draw,minimum size=0.7cm}]

  \node[main node] (1) {1};
  \node[main node] (3) [below right=of 1] {3};
  \node[spe node] (2) [above right =of 3] {\textcolor{white}{2}};
  \node[main node] (4) [below left=of 3] {4};
  \node[main node] (5) [below right=of 3] {5};
  \node[spe node] (6)  [right= of 3] {\textcolor{white}{6}};

  \path[every node/.style={font=\sffamily\small}]
    (2) edge node {}  (6);
\end{tikzpicture}\\ &&& \\ \hline
\end{tabular}
}
\caption{Example of an {\MCN} game with unitary costs and benefits, and budgets $\Omega=\Phi=\Lambda=1$. We removed the vaccinated and protected vertices as an infection cannot pass through them (see Property~\ref{prop_node_delet}). Vertices $\{1,3,4,5\}$ are saved and $\{2,6\}$ are infected.}
\label{fig:example}
\end{figure}
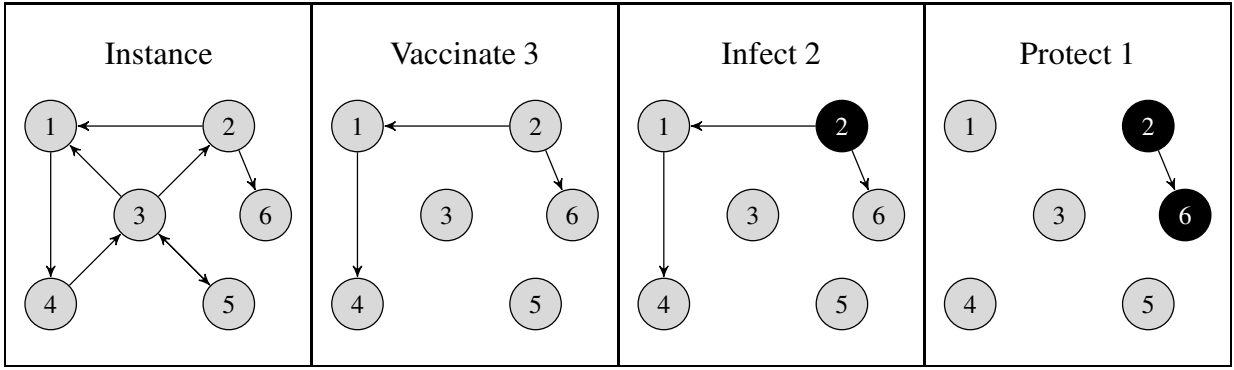

%% file: Paper_Version/Undirected_Graphs_unitary.tex
\section{Undirected graphs: the unitary case}\label{sec:undirected_uni}

In this section, we focus on undirected graphs $G=(V, E)$, \ie, for each couple of vertices $(u,v) \in V \times V$, if the arc $(u,v)$ is in $G$, then $(v,u)$ is also in the graph. We thus call $E$ the set of  edges. Here, we also consider unit benefits and costs, \ie, $\forall v \in V, \; \hat{c}_v = h_v = c_v = b_v = 1$. We introduce $s$, the function that, given a graph $G$, the vaccination strategy $D$, the attack strategy $I$ and the protection strategy $P$, returns $s(G, D, I, P)$, the number of saved vertices in the end of the game. Thus, in this setting, the trilevel formulation of the problem is simply:
\begin{equation}
    \underset{ \substack{D \subseteq V \\ |D| \leq \Omega}}{\max} \quad \underset{ \substack{I \subseteq V \backslash D \\ |I| \leq \Phi}}{\min} \quad \underset{ \substack{P \subseteq V \backslash (I \cup D) \\ |P| \leq \Lambda}}{\max} s(G, D, I, P).
    \label{MCN_undir}
\end{equation}
To ease our analysis, guided by the relationship between \textit{Critical Node Detection Problems} and \textit{Node-Deletion Problems}, we first write the immediate \hyperref[prop_node_delet]{Property \ref{prop_node_delet}}\footnote{It is easy to see that Property~\ref{prop_node_delet} holds for general directed weighted versions with $s(G, D, I, P) = s(G[V\backslash (D \cup P)], \emptyset, I, \emptyset) + \sum_{v \in D} b_v + \sum_{v \in P} b_v$ and  $s(G, D, I, P)$ equal to the benefit associated with the saved vertices in the end of the game.} stating that \textit{vaccinating} or \textit{protecting} vertices has the same effect as \textit{removing} them from the graph with respect to $s$. Starting from $G=(V, E)$ and a subset $W \subseteq V$, we denote by $G[V\backslash W]$ the graph induced by the deletion of the vertices in $W$ and its incident edges.
\begin{prop}\label{prop_node_delet}
Given $G, D, I, P$, we have that $$s(G, D, I, P) = s(G[V\backslash (D \cup P)], \emptyset, I, \emptyset) + |D| + |P|.$$
\end{prop}
What \hyperref[prop_node_delet]{Property \ref{prop_node_delet}} actually says is that the infected vertices in $G$ are the ones in the connected components of $G[V\backslash (D \cup P)]$ where there is at least one attacked vertex in $I$.

We will start by classifying the computational complexity of {\sc Protect}, followed by the one of {\sc Attack-Protect}, and, finally, {\sc Vaccination-Attack}. From the latter, we obtain the complexity of {\sc Attack}, and the minimum complexity of {\MCN}.

\subsection{The {\sc Protection} problem}

In {\sc Protect}, the defender is given $D$ and $I$ and seeks to find an optimal $P$. Thus, thanks to \hyperref[prop_node_delet]{Property \ref{prop_node_delet}}, we can assume that the game takes place in $G_a = G[V \backslash D]$ for this last move: the defender wants to find at most $\Lambda$ vertices $P \subseteq V_a \backslash I$ that will maximixe $s(G_a, \emptyset, I, P)$. For a given choice of $P$, we introduce $C_1(P), ..., C_{N(P)}(P)$, the $N(P)$ connected components in the graph $G_a[V_a \backslash P]$. Hence, the objective of the defender being to find $P$ minimizing the number of infected vertices $f(P)$, we can define it as:
\begin{equation}
    f(P) = \sum_{i=1}^{N(P)}|C_i(P)| \times \mathbbm{1}_{C_i(P) \cap I \neq \emptyset}.
    \label{obj_protect}
\end{equation}
We will show that finding such a $P$ is NP-complete. We argue that it is a direct consequence of the results of \cite{Addis2013} showing that the \textit{Critical Node Problem} is NP-hard on split graphs.

\subsubsection{The Critical Node Problem on split graphs}

The \textit{Critical Node Problem} \textit{(CNP)} \cite{ARULSELVAN20092193} is a related problem to ours. The setting is very similar to {\sc Protection}: we have an undirected graph $\bar{G}=(\bar{V},\bar{E})$, an integer budget $B$, and we want to find a subset $\bar{P}$ of vertices to remove that minimizes the \textit{pairwise connectivity} of the residual subgraph $\bar{G}[\bar{V}\backslash \bar{P}]$ under the constraint of having $|\bar{P}| \leq B$. If we denote by $\bar{C_1}(\bar{P}), ..., \bar{C}_{N(\bar{P})}(\bar{P})$ the $N(\bar{P})$ connected components of $\bar{G}[\bar{V}\backslash \bar{P}]$, the measure we want to minimize is:
\begin{equation}
    g(\bar{P}) = \sum_{i=1}^{N(\bar{P})} \binom{|\bar{C_i}(\bar{P})|}{2}
    \label{pairwise_connectivity}
\end{equation}
where each term in the sum is the pairwise connectivity of $\bar{C_i}(\bar{P})$. Here, we will focus more particularly on \textit{split graphs}. A split graph is a graph $\bar{G}=(\bar{V}, \bar{E})$ whose vertices $\bar{V}$ can be split in two sets $\bar{V_1}$ and $\bar{V_2}$, $\bar{V_1}$ forming a clique and $\bar{V_2}$ an independent set. Thus, the decision problem for this particular case of the CNP is:

\vspace{0.3cm}
\boxxx{
{\bf CNP}$_{split}$: \\
{\sc instance}: A split graph $\bar{G}=(\bar{V_1}, \bar{V_2}; \bar{E})$, a non-negative integer budget $B \leq |\bar{V}|$ and a non-negative integer $\bar{K}$. \\
{\sc question}: Is there a subset $\bar{P} \subseteq \bar{V}$, $\bar{P} \leq B$ such that $g(\bar{P}) \leq \bar{K}$?
}
\vspace{0.3cm}

As \cite{Addis2013} noted, in this setting there is at most one connected component of the residual subgraph $\bar{G}[\bar{V}\backslash \bar{P}]$ that contains more than one vertex. Moreover, it is easy to see that if this \textit{nontrivial connected component} exists, it necessarily contains a subclique of $\bar{G}[\bar{V_1}]$. More than that, it is the only connected component of $\bar{G}[\bar{V}\backslash \bar{P}]$ containing vertices from $\bar{V_1}$. Thus, we can name $\bar{C_1}$ the connected component containing vertices of $\bar{V_1}$ \textit{(in the case of $\bar{P} \supseteq \bar{V_1}$, then $\bar{C_1}$ is either a singleton from $\bar{V_2}$ or is empty and our reasoning still holds)}. Then, minimizing (\ref{pairwise_connectivity}) is equivalent to minimize $|\bar{C_1}|$. But finding the subset of vertices $\bar{P}$ to remove to do that has been shown to be NP-hard:

\begin{lem}{\cite{Addis2013}} {\sc CNP}$_{split}$ is NP-hard. \end{lem}

\subsubsection{Complexity result}
Next, we show that the decison version of {\sc Protect} is NP-complete using a reduction from CNP$_{split}$. The decision problem is the following:

\vspace{0.3cm}
\boxxx{
 \textsc{\textbf{Protect}}: \\
{\sc instance}: A graph $G_a=(V_a,E_a)$, a set of attacked vertices $I \subseteq V_a$, a non-negative integer budget $\Lambda \leq |V_a|-|I|$ and a non-negative integer $K$.
\\
{\sc question}: Is there a subset $P \subseteq V_a \backslash I$, $|P| \leq \Lambda$ such that the number of infected vertices $f(P) \leq K$?
}
\vspace{0.3cm}

Note that the question can be equivalently re-written with the inequality $s(G_a, \emptyset,I,P) \geq \vert V_a \vert -K$.

\begin{thm}
{\sc Protect} is NP-complete.
\label{th_protect_NP}
\end{thm}
\begin{proof}
It is easy to see that {\sc Protect} is  NP as determining the objective value only requires finding the connected components of $G_a[V_a \backslash P]$ which can be done in linear time using a  depth-first search (DFS).
\\
To complete the proof, we exhibit an immediate reduction from {\sc CNP}$_{split}$. Let us take an instance of this problem, \ie ~a split graph $\bar{G}=(\bar{V_1}, \bar{V_2}; \bar{E})$, a non-negative integer budget $B$ and a non-negative integer $\bar{K}$. Given that, we build a graph $G_a$ by growing by one the size of the clique $\bar{G}[\bar{V_1}]$ with the addition of a vertex $u$. Thus, $V_a = \bar{V_1} \cup \{u\} \cup \bar{V_2}$ and $E_a$ is obtained by taking $\bar{E}$ and adding an edge $(u,\bar{v_1}) \;  \forall \bar{v_1} \in \bar{V_1}$. In fact, the new graph is still a split graph $G_a = (\bar{V_1} \cup \{u\}, \bar{V_2};E_a)$. Finally, the corresponding instance of {\sc Protect} is given by $G_a$, $I=\{u\}$, $\Lambda = B$ and $K= \big \lfloor \frac{1}{2}(3+\sqrt{\vphantom{\prod^2}8 \bar{K}+1}) \big \rfloor$ (obtained by solving $\bar{K} = \binom{K-1}{2}$). An example of such construction can be found in \hyperref[Fig1]{Figure \ref{Fig1}}. Then, as there is only one attacked vertex, minimizing (\ref{obj_protect}) on this instance of {\sc Protect} corresponds to choosing a $P$ that minimizes the size of the unique connected component to which $u$ belongs in $G_a[V_a \backslash P]$. Let's name $C_1$ this connected component. But as $u$ belongs to the clique part of the split graph $G_a$, $C_1$ is also the unique connected component of $G_a[V_a \backslash P]$ containing vertices from $V_1 = \bar{V_1} \cup \{u\}$. Thus, we have that $C_1= \bar{C_1} \cup \{u\}$ and $g(P)= \begin{pmatrix} f(P) - 1 \\ 2 \end{pmatrix}$. Hence, finding $P$ that minimizes $f$ on $G_a$ is equivalent to finding $P$ that minimizes $g$ on $\bar{G}$. This finishes the proof that {\sc Protect} is NP-hard.
\end{proof}

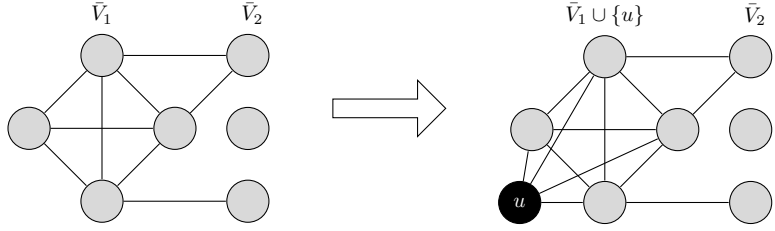
\begin{figure}[ht]
    \centering
    \input{Paper_Version/Figure1}
    \caption{Example of construction of $G_a$ from $\bar{G}$}
    \label{Fig1}
\end{figure}

\begin{rem}
In \cite{Barnetson2020}, it shown that the {\sc Firebreak} problem is NP-complete. This problem is equivalent to {\sc Protect} when $\vert I\vert =1$. Hence, their result can be used to establish Theorem~\ref{th_protect_NP}. Nevertheless, given that our reductions differ significantly and were obtain independently, we decided to present our alternative reduction.
\end{rem}

\subsection{The {\sc Attack-Protect} problem}

We showed that solving the last level of MCN is NP-complete, now we will prove that {\sc Attack-Protect} is also NP-hard. In this bilevel problem, we are taking the side of the attacker: the aim is to find the attack that will maximize the number of infected vertices after protection. The decision version of the problem is:

\vspace{0.3cm}
\boxxx{
\textsc{\textbf{Attack-Protect}}: \\
{\sc instance}: A graph $G_a=(V_a, E_a)$, two non-negative integer budgets $\Phi, \Lambda$ such that $\Phi + \Lambda \leq |V_a|$ and a non-negative integer $K \leq |V_a|$ \\
{\sc question}:  Is there a subset $I \subseteq V_a$, $|I| = \Phi$ such that $\forall P \subseteq V_a \backslash I$, $|P| \leq \Lambda$, the number of infected vertices $f(P) \geq K$?
}
\vspace{0,3cm}

We will use a reduction from the \textit{Dominating Set} problem, a known NP-complete problem \cite{NPBook}, whose decision version is:

\vspace{0.3cm}
\boxxx{
\textsc{\textbf{Dominating Set}}: \\
{\sc instance}: A graph $\bar{G}=(\bar{V}, \bar{E})$, a positive integer $B \leq |\bar{V}|$ \\
{\sc question}: Is there a subset $U \subseteq \bar{V}$, $|U| \leq B$, such that $\forall v \in \bar{V} \backslash U$, $\exists \; u \in U$ such that $(u,v) \in \bar{E}$?
}
\vspace{0.3cm}

\begin{thm}
{\sc Attack-Protect} is NP-hard.
\label{th_AP_NP}
\end{thm}

\begin{proof}
Let us take a graph $\bar{G}=(\bar{V}, \bar{E})$ and a positive integer $B \leq |\bar{V}|$. The instance of {\sc Attack-Protect} is simply created by taking $G_a=\bar{G}$, $\Phi = B$, $\Lambda = |V_a| - \Phi - 1$ and $K = \Phi +1$. In this configuration, we have a protection budget $\Lambda$ which is exactly one less than the number of vertices that are not attacked. Thus, if all the protection budget is spent, there is only one vertex $u$ in the graph that is neither attacked nor protected. Therefore, if $u$ becomes infected after protection \textit{(i.e $f(P) = K = \Phi +1$)}, that means that the protection strategy did not manage to save one unit of budget while saving all the other vertices, meaning that the other vertices were all in direct contact with at least one attacked one \textit{(if it was not the case, one unit of budget could have been saved by protecting all the neighbors of the vertex that is not in direct contact with $I$)}. As $u$ also becomes infected, it also means that it is adjacent to one vertex in $I$. Thus, finding $I$ such that $\forall P, f(P) \geq K$ means that $I$ is a dominating set of size $B$,  which concludes the proof.
\end{proof}

\subsection{The {\sc Vaccination-Attack} problem}

In this part, we will ignore the fact that there is a protection stage at the end. This is a particular case of {\MCN} since it is equivalent to studying it  with protection budget $\Lambda=0$. We will show that the bilevel problem {\sc Vaccination-Attack} is NP-complete. The decision problem is the following:

\vspace{0.3cm}
\boxxx{
\textsc{\textbf{Vaccination-Attack}}: \\
{\sc instance}: A graph $G=(V,E)$, two non-negative integer budgets $\Omega$ and $\Phi$ such that $\Omega + \Phi \leq |V|$ and a non-negative integer $K$.
\\
{\sc question}: Is there a subset $D \subseteq V$, $|D| \leq \Omega$ such that $\forall I \subseteq V\backslash D$ with $|I| \leq \Phi$, the number of infected vertices $|V| - s(G, D, I, \emptyset)\leq K$?
}
\vspace{0.3cm}

First, we argue that in this configuration, finding the optimal attack following a given vaccination can be done in polynomial time.
\begin{lem}
{\sc Vaccination-Attack} $\in$ NP. Moreover, {\sc Attack} can be solved in polynomial time.
\label{th_VA_NP}
\end{lem}
\begin{proof}
Given a vaccinated set $D$, we want to verify that all the possible subsequent attacks cannot infect more than $K$ vertices. To do that, it suffices to find the best attack, \ie, solve the Attacker optimization problem, and check whether or not it complies with the inequality. But, as we highlighted it with \hyperref[prop_node_delet]{Property \ref{prop_node_delet}}, the graph on which the attack phase takes place is $G_a = G[V \backslash D]$ and the saved vertices in the end are exactly the ones in the connected components of $G_a$ that do not contain any attacked vertex. Thus, the best attack possible given $G_a$ and budget $\Phi$ is to infect one vertex in each of the $\Phi$ largest connected components of $G_a$. This can be done in linear time using a DFS. Hence, {\sc Vaccination-Attack} $\in$ NP.
\end{proof}
In fact, this proof showed that {\sc Vaccination-Attack} is actually equivalent to another problem: finding a subset of vertices $D$ to remove from $G$ that minimizes the sum of the sizes of the $\Phi$ largest connected components in the induced subgraph. Let's call this problem \textsc{MinMax$\Phi$C}:

\vspace{0.3cm}
\boxxx{
\textbf{\textsc{MinMax$\Phi$C}}: \\
{\sc instance}: A graph $G=(V,E)$, two non-negative integer budgets $\Omega$ and $\Phi$ such that $\Omega + \Phi \leq |V|$ and a non-negative integer $K$.
\\
{\sc question}: Is there a subset $D \subseteq V$, $|D| \leq \Omega$ such that the sum of the sizes of the $\Phi$ largest connected component in $G[V \backslash D]$ is less than $K$?
}
\vspace{0.3cm}

\begin{lem}
{\sc Vaccination-Attack} and \textsc{MinMax$\Phi$C} are equivalent problems.
\end{lem}
As Shen \etal ~\cite{Shen2012} argued that \textsc{MinMaxC}, the problem that only seeks to minimize the size of the largest connected component in the residual graph, is NP-hard,  we have as a direct consequence that \textsc{MinMax$\Phi$C} is also NP-hard, which leads to the following corollaries:
\begin{cor}
{\sc Vaccination-Attack} is NP-complete.
\label{cor_VA_NP}
\end{cor}
\begin{cor}
{\sc MCN} is NP-hard.
\end{cor}
\begin{proof}
Given an instance of {\sc Vaccination-Attack}, there is a corresponding instance of MCN by taking the same $G, \Omega, \Phi, K$ and by setting $\Lambda = 0$.
\end{proof}

%% file: Paper_Version/Figure1.tex
\begin{tabular}{ccc}
\begin{tikzpicture}[-,>=stealth',shorten >=0.2pt,auto,node distance=1cm,baseline={(current bounding box.north)},  scale=0.8, transform shape,
  main node/.style={circle,fill=black!15,draw,minimum size=0.7cm},spe node/.style={fill=white,draw=white}]

  \node[main node] (1) {};
  \node[main node] (2) [below left=of 1] {};
  \node[main node] (4) [below right =of 2] {};
  \node[main node] (3) [below right=of 1] {};
  \node[main node] (5) [above right=of 3] {};
  \node[main node] (6)  [below=0.5cm of 5] {};
  \node[main node] (7) [below right=of 3] {};
  \node[spe node] (8) at (0,0.7) {$\bar{V}_1$};
  \node[spe node] (9) at (2.5,0.7) {$\bar{V}_2$};

  \path[every node/.style={font=\sffamily\small}]
    (1) edge node {}  (2)
        edge node {}  (3)
        edge node {}  (4)
        edge node {}  (5)
    (2) edge node {}  (3)
        edge node {}  (4)
    (3) edge node {}  (4)
        edge node {}  (5)
    (4) edge node {}  (7);
\end{tikzpicture}
& \hspace{1cm}
  \begin{tikzpicture}
    \useasboundingbox (-2,0);
    \node[single arrow,draw=black,fill=white,minimum height=1.5cm,shape border rotate=0] at (-2,-1.5) {};
  \end{tikzpicture}
  \hspace{1cm}
&
\begin{tikzpicture}[-,>=stealth',shorten >=0.2pt,auto,node distance=1cm,baseline={(current bounding box.north)},  scale=0.8, transform shape,
  main node/.style={circle,fill=black!15,draw,minimum size=0.7cm},spe node/.style={fill=white,draw=white}, inf node/.style={fill=black,draw,circle,minimum size=0.7cm}]

  \node[main node] (1) {};
  \node[main node] (2) [below left=of 1] {};
  \node[main node] (4) [below right =of 2] {};
  \node[main node] (3) [below right=of 1] {};
  \node[main node] (5) [above right=of 3] {};
  \node[main node] (6)  [below=0.5cm of 5] {};
  \node[main node] (7) [below right=of 3] {};
  \node[spe node] (8) at (0,0.7) {$\bar{V}_1 \cup \lbrace u \rbrace$};
  \node[spe node] (9) at (2.5,0.7) {$\bar{V}_2$};
   \node[inf node] (10) [left=0.7cm of 4] {\textcolor{white}{$u$}};

  \path[every node/.style={font=\sffamily\small}]
    (1) edge node {}  (2)
        edge node {}  (3)
        edge node {}  (4)
        edge node {}  (5)
    (2) edge node {}  (3)
        edge node {}  (4)
    (3) edge node {}  (4)
        edge node {}  (5)
    (4) edge node {}  (7)
    (10) edge node {} (1)
         edge node {} (2)
         edge node {} (3)
         edge node {} (4);
\end{tikzpicture}
\end{tabular}

%% file: Paper_Version/Undirected_Graphs_weighted_adel_margarida.tex
\section{Undirected graphs: the weighted case}
\label{sec:undirected_wei}

In this section, we study the version of MCN presented in problem\hyperref[MCN_formulation]{~\eqref{MCN_formulation}} restricted to undirected graphs. We will use the subscript $w$ to denote the weighted version, MCN$_w$, as well as for its subgames. In this problem, given a graph $G=(V,E)$, each vertex $v \in V$ is associated with a benefit $b_v$ and  cost parameters $\hat{c}_v, h_v$ and $c_v$, respectively the cost of vaccinating, attacking and protecting vertex $v$. First, note that the NP-completeness of {\sc Protect$_w$} is immediate from the previous section. 

Having introduced costs and benefits, our game and its subgames are intimately related to \textit{Knapsack problems}, which we will use to demonstrate all of our complexity results in this part.  We will start by highlighting the direct relationship between {\sc Attack}$_w$ and {\sc Knapsack}, which will get us the NP-completeness of this problem. Then, we will focus on the two bilevel sub-problems {\sc Vaccination-Attack}$_w$ and {\sc Attack-Protect}$_w$ and prove they are $\Sigma_2^p$-complete thanks to a \textit{Knapsack Interdiction problem}. To conclude, we show that MCN$_w$ is $\Sigma_3^p$-complete. We will observe that the introduction of non-unitary parameters offers sufficient flexibility to go a level up in the polynomial hierarchy in comparison with the unitary undirected cases.

\subsection{The {\sc Attack$_w$} problem}

In the attack phase, the vaccination already took place so we effectively work on $G_a$, which is the result of the deletion of the vaccinated vertices from the original graph. We are given a non-negative attack budget $\Phi$, and as there is no protection phase afterwards, we set  $\Lambda = 0$. The goal is thus to harvest the most benefit possible by infecting vertices subject to a budget limit. The decision version of the problem is then:

\vspace{0.3cm}
\boxxx{
\textbf{\textsc{Attack$_w$}}: \\
{\sc instance}: An undirected graph $G_a = (V_a, E_a)$, a non-negative integer cost $h_v$ and value $b_v$ for each vertex $v \in V$, a non-negative integer budget $\Phi$, and a non-negative integer number $K$. \\
{\sc question}: Is there a subset of vertices $I \subseteq V_a$ to attack, with cost $\sum_{v \in I} h_v \leq \Phi$ such that the sum of the benefits of the resulting infected vertices in $G_a$ is greater or equal to $K$?
}
\vspace{0.3cm}

To make evident the NP-completeness of the problem, we simply state the decision version of the \textit{Knapsack problem}, one of the Karp's $21$ NP-complete problems \cite{Karp1972}: 

\vspace{0.3cm}
\boxxx{
\textbf{\textsc{Knapsack}}: \\
{\sc instance}: Finite set $U$, for each $u \in U$, a positive integer size $a_u \in \mathbb{N}$ and a positive integer profit $p_u \in \mathbb{N}$, and two positive integers $B$ and $\bar{K}$. \\
{\sc question}: Is there a subset $U' \subseteq U$ such that $\sum_{u \in U'} a_u \leq B$ verifying $\sum_{u \in U'} p_u \geq \bar{K}$?
}
\vspace{0.3cm}

\begin{thm}
\textsc{Attack$_w$} is  equivalent to \textsc{Knapsack}.
\label{Attack_w_KP}
\end{thm}
\begin{proof}
First, we prove that each instance of \textsc{Attack$_w$} reduces to an instance of \textsc{Knapsack}. Given an instance of \textsc{Attack$_w$}, it is straightforward to see that it is sufficient to infect the vertex $v$ with lowest infection cost $h_v$ of a given  connected component to infect the whole component and collect the benefit $b$ of each vertex included in that component. If $N(G_a)$ represents the set of  connected components of $G_a$, to each connected component $C\in N(G_a)$ we can  assign a total profit $b_C=\sum_{v\in C} b_v$ and infection cost $h_C=\min_{v\in C}h_v$. We can then straightforwardly build a \textsc{Knapsack} instance where the set  $N(G_a)$ is mapped to $U$, $a_u=h_C$  and $p_u=b_C$  for $C \in N(G_a)$, and $B=\Phi$ and $\bar{K}=K$. 

Conversely, if we start from an instance of \textsc{Knapsack}, we construct an instance of \textsc{Attack$_w$} by setting $V_a = U$, $E_a = \emptyset$, $K=\bar{K}$, $\Phi=B$ , and $\forall v \in V_a$, $h_v=a_v$, $b_v=p_v$. In this configuration, $G_a$ having no edges, the attacked vertices are exactly the infected ones in the end, and the goal of the attacker is equivalent to filling up a knapsack with limited capacity by choosing which vertices to attack.\\
Given that both \textsc{Attack$_w$}  and \textsc{Knapsack} can be reduced to each other, both problems are equivalent.
\end{proof}

Remark that given an attack $I$, finding the subsequent infected vertices can be done in linear time thanks to a DFS. Then, it suffices to sum the cost of the vertices in $I$ to verify the budget constraints and to sum the benefits associated with the infected vertices to verify that it is greater or equal to $K$. Hence, \textsc{Attack$_w$} $\in$ NP and thus: 

\begin{cor}
\textsc{Attack$_w$} on undirected graphs in weakly NP-complete, even on trivial graphs.
\end{cor}
\begin{proof}
Since it well known that \textsc{Knapsack} is weakly NP-complete, the result follows from the above theorem. Moreover, since any instance of \textsc{Knapsack} reduces to an instance of \textsc{Attack$_w$} which has no edges, \textsc{Attack$_w$} is NP-complete on trivial graphs.
\end{proof}


\subsection{The {\sc Attack-Protect}$_w$ problem}\label{subsec:undirected_wei_AP}


In the proof of \hyperref[Attack_w_KP]{Theorem \ref{Attack_w_KP}}, we highlighted how a {\sc Knapsack} instance can be directly transformed into a weighted graph with no edges. In this section, as well as in the next one, we will use a similar transformation, but add one additional \textit{root vertex} to our construction in order to build a star graph: one root vertex connected with an edge to each of the other vertices, each one representing an item $u \in U$ of the knapsack. That way, the complexity results we devise hold for trees, a very particular class of graphs where frequently  theoretically intractable problems become polynomially solvable.  

As before, the vaccination having already been done, we start from $G_a$, the graph where the vaccinated vertices have been removed.

\vspace{0.3cm}
\boxxx{
\textsc{\textbf{Attack-Protect}$_w$}: \\
{\sc instance}: A graph $G_a=(V_a,E_a)$, a non-negative integer $K$, two non-negative integer budgets $\Phi$ and $\Lambda$, $\forall v \in V_a$ two non-negative integer costs $h_v$, $c_v$ and a non-negative integer benefit $b_v$.
\\
{\sc question}: Is there a subset $I \subseteq V_a$, with cost $\sum_{v \in I} h_v \leq \Phi$ such that $\forall P \subseteq V_a\backslash I$ with cost $\sum_{v \in P}c_v \leq \Lambda$, the sum of the benefit of the saved vertices  is strictly less than $K$?
}
\vspace{0.3cm}

In order to show that {\sc Attack-Protect}$_w$ is $\Sigma_2^p$-complete, we use the \textit{Bilevel Interdiction Knapsack Problem} introduced by DeNegre \cite{DeNegre2011} and proven to be $\Sigma_2^p$-complete in \cite{Caprara2014}. In this problem, two players, a leader and a follower, can select items in the same set of objects $O$. First, the leader packs some items into her knapsack, then the follower chooses among the remaining ones. The aim of the leader is to interdict a subset of items, subject to a capacity constraint, in order to minimize the total profit of the follower. The objective of the follower is to maximize her profit, subject to a constraint capping the maximum profit obtainable by her. The decision problem is then:

\vspace{0.3cm}
\boxxx{
\textsc{\textbf{Bilevel Interdiction Knapsack (BIK)}}: \\
{\sc instance}: A set of items $O$ such that each $o \in O$ has a positive integer weight $a_o$ and a positive integer profit $p_o$, a positive integer maximum weight capacity $A$ for the leader, a positive integer maximum profit $B$ for the follower, and a positive integer $\bar{K} \leq B$.\\
{\sc question}: Is there a subset $O_l\subseteq O$ of items for the leader to select, with $\sum_{o \in O_l} a_o \leq A$, such that every subset $O_f\subseteq O\setminus O_l$ with $\sum_{o \in O_f} p_o \leq B$ that the follower can create has a total profit $\sum_{o \in O_f} p_o < \bar{K}$?
}
\vspace{0.3cm}

\begin{thm}
{\sc Attack-Protect}$_w$ is strongly $\Sigma_2^p$-complete, even if the graph is a tree.
\label{th_AP_w}
\end{thm}
\begin{proof}
First, {\sc Attack-Protect}$_w$ is in  $\Sigma_2^p$ since   this decision problem is exactly of the form $\exists I \ \ \forall P \ \ Q(I,P)$, where $Q(I,P)$ is a proposition that can be evaluated in polynomial time (\ie, it verifies the attack and protection budget constraints, as well as, the benefit of the saved vertices).

Next, we prove the problem $\Sigma_2^p$-hardness. Let us begin by noting that we can restrict the instances of KIP to the ones where $\bar{K}$ and $B$ are strictly inferior to $\sum_{o \in O} p_o$, otherwise, KIP reduces to {\sc Knapsack}. This remark is used in the second part of this proof.

Starting from an instance of BIK, we construct an instance of {\sc Attack-Protect$_w$} as  follows. We first build a star graph $G_a=(V_a,E_a)$ with a root vertex $r$ and a vertex $v_o$ for each $o\in O$ linked to $r$ through an edge $(r,v_o)$. We set $b_r=\sum_{o \in O} p_o + 1$ and $h_r=c_r=1$. We also set $b_{v_o}=c_{v_o}=p_o$ and $h_{v_o}=a_o$ for each $o\in O$. See Figure~\ref{Fig:weight_AP}. Finally, we set $\Phi=A+1$, $\Lambda=B$ and $K= \bar{K}$. 


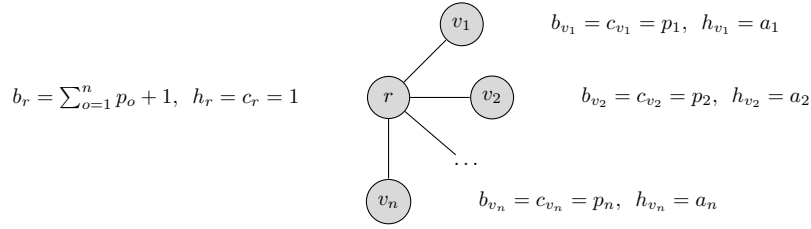
\begin{figure}[ht]
    \centering
    \input{Paper_Version/Figure_weight_2}
    \caption{Graph reduction from BIK to {\sc Attack-Protect}$_w$ when $O=\lbrace 1,2 \ldots, n \rbrace$.}
    \label{Fig:weight_AP}
\end{figure}

Suppose first that  BIK is a \emph{Yes} instance. Then, there is a set of items $O_l\subseteq O$ of total weight $\sum_{o \in O_l} a_o\leq A$ such that for all $O_f \subseteq O \setminus O_l$ feasible for the follower, it holds $\sum_{o \in O_f} p_o \leq \bar{K}-1$. Consequently, in the {\sc Attack-Protect$_w$}, the attacker can select the subset of vertices $I=\{r\}\cup\{v_o:\ o\in O_l\}$ with a feasible attacking cost $ \sum_{v \in I} h_v=1+\sum_{o \in O_l}a_o\leq A+1=\Phi$. Now, the defender can only protect vertices in $\{v_o:\ o\notin O_l\}$ and since the central vertex of the star graph is infected, the saved vertices will be the protected ones. The aim of the defender is therefore to select the subset of vertices of maximum total benefit with respect to the protection budget $\Lambda$. This is exactly the follower's problem in BIK. Hence, since BIK is an \emph{Yes} instance, the defender (follower in BIK) cannot attain a benefit (profit in BIK) equal or greater to $K=\bar{K}$ through a feasible action. Therefore, the {\sc Attack-Protect$_w$} is a \emph{Yes} instance.


Now suppose that {\sc Attack-Protection$_w$} is a \emph{Yes} instance. Thus, there exists an attack strategy $I\subseteq V_a$ such that there is \textbf{no} feasible subset $P\subseteq V_a\setminus I$ of protected vertices leading to a total benefit greater or equal to $K$ for the defender. 
As $\Phi \geq 1$, it is obvious that the attacker  will attack at least the central vertex $r$, otherwise, the defender would pick it and achieve a benefit superior to $K$ (recall that $K=\bar{K}<  \sum_{o \in O} p_o$), contradicting  {\sc Attack-Protection$_w$} \emph{Yes} instance. Hence, the attacker is left with budget $\Phi-h_r=A$. Once the central vertex is attacked, only the other vertices subsequently protected will not be infected. Therefore, the rest of the attack budget $A$ is spent on a subset of vertices of $\{v_o\in V_a:\ o\in O\}$ and it ensures that for any $P=\{v_o\in V_a:\ o\in O\setminus I\}$ with $\sum_{v \in P} c_v=\sum_{o: v_o \in P} p_v \leq \Lambda = B$, the total benefit for the defender is $\sum_{v \in P} b_v= \sum_{o: v_o \in P} p_v \leq \bar{K}-1$.  Consequently, BIK is also a \emph{Yes} instance.

This completes the proof that {\sc Attack-Protect$_w$} is $\Sigma_2^p$-complete. Moreover, since the BIK was shown to be NP-complete even for unary encoding, we can conclude that no pseudopolynomial-time algorithm exists to solve the {\sc Attack-Protect} subgame. Since a star graph is a tree, the result stated in the theorem holds.
\end{proof}

\subsection{The {\sc Vaccination-Attack$_w$} problem}\label{subsec:undirected_wei_VA}

Using a similar reduction to the one in the proof of \hyperref[th_AP_w]{Theorem \ref{th_AP_w}}, we show that the {\sc Vaccination-Attack$_w$} on weighted graphs is $\Sigma_2^p$-complete. As in the unitary case, this is equivalent to studying MCN$_w$ problems where we set $\Lambda=0$. The decision version of the problem is:


\vspace{0.3cm}
\boxxx{
\textsc{\textbf{Vaccination-Attack}$_w$}: \\
{\sc instance}: A graph $G=(V,E)$, a non-negative integer $K$, two non-negative integer budgets $\Omega$ and $\Phi$, $\forall v \in V$ two non-negative integer costs $\hat{c}_v$, $h_v$ and a non-negative integer benefit $b_v$.
\\
{\sc question}: Is there a subset $D \subseteq V$, with cost $\sum_{v \in D}\hat{c}_v \leq \Omega$ such that $\forall I \subseteq V\backslash D$ with cost $\sum_{v \in I} h_v \leq \Phi$, the sum of the benefit of the infected vertices  is strictly less than $K$?
}
\vspace{0.3cm}



\begin{thm}
{\sc Vaccination-Attack}$_w$ is strongly $\Sigma_2^p$-complete, even if the graph is a tree.
\label{th_VA_w}
\end{thm}
\begin{proof}
As before, {\sc Vaccination-Attack}$_w$ is in  $\Sigma_2^p$ since   this decision problem is exactly of the form $\exists D \ \ \forall I \ \ Q(D,I)$  is a proposition that can be evaluated in polynomial time.

Now, we establish the problem $\Sigma_2^p$-hardness.
We start from an instance of BIK, defined in the previous section, and we then construct an instance of  {\sc Vaccination-Attack$_w$} as follows. First, we build a star graph $G=(V,E)$ with a central vertex $r$ and $|O|$ leaf vertices $v_o$ with $o\in O$. See Figure~\ref{Fig:weight_VA}. We add an edge $(r,v_o)$ for each such leaf vertex. The central vertex has benefit $b_{r}=\bar{K}$ and costs $\hat c_{r}=h_{r}=1$. Each leaf vertex $v_o$ with $o\in O$ has a benefit $b_{v_o}=p_o$, cost for the defender $\hat c_{v_o}=a_o$ and cost for the attacker $h_{v_o}=p_o$. Finally, we fix $\Omega=A+1$, $\Phi=B$ and $K=\bar{K}$.

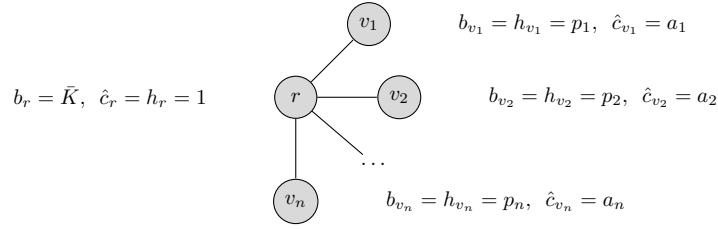
\begin{figure}[ht]
    \centering
    \input{Paper_Version/Figure_weight_1}
    \caption{Graph reduction from BIK to {\sc Vaccination-Attack}$_w$ when $O=\lbrace 1,2 \ldots, n \rbrace$.}
    \label{Fig:weight_VA}
\end{figure}

This is exactly the setting of BIK and one can easily complete the proof of equivalence of the two decision instances following a path very similar to the proof of \hyperref[th_AP_w]{Theorem \ref{th_AP_w}}.

Finally, the reduction used a star graph which is a particular case of a tree. Hence, the problem is $\Sigma_2^p$-complete even on trees.
\end{proof}


\subsection{The MCN$_w$ problem}

In this section we show that the decision problem  MCN$_w$  is $\Sigma_3^p$-complete.

\vspace{0.3cm}
\boxxx{
\textsc{\textbf{MCN}$_w$}: \\
{\sc instance}: A graph $G=(V,E)$, a non-negative integer $K$, three non-negative integer budgets $\Omega$, $\Phi$ and $\Lambda$, $\forall v \in V$ three non-negative integer costs $\hat{c}_v$, $h_v$  and $c_v$, and a non-negative integer benefit $b_v$.
\\
{\sc question}: Is there a subset $D \subseteq V$, with cost $\sum_{v \in D} \hat{c}_v \leq \Omega$ such that $\forall I \subseteq V\backslash D$ with cost $\sum_{v \in I} h_v \leq \Phi$, there is $P \subseteq V \backslash I$ with cost  $\sum_{v \in D} c_v \leq \Lambda$ such that the sum of the benefit of the saved  vertices  is greater or equal to $K$?
}
\vspace{0.3cm}


In order to achieve our ultimate goal, we take the \emph{3-Alternating Quantified Satisfiability} problem ($B_3 \cap 3CNF$), known to be $\Sigma_3^p$-complete problem~\cite{Meyer1973,Wrathall1976}, in order to prove that the generalization of BIK to a trilevel, the \emph{Trilevel Interdiction Knapsack} (TIK), is $\Sigma_3^p$-complete. Then, TIK is used to demonstrate that MCN$_w$ is $\Sigma_3^p$-complete.

\vspace{0.3cm}
\boxxx{
\textbf{\textsc{3-Alternating Quantified Satisfiability}} ($B_3 \cap 3CNF$):\\
{\sc instance}: Disjoint non-empty sets of variables $X$, $Y$ and $Z$, and a Boolean expression  $E$ over $U = X \cup Y \cup Z$ in conjunctive normal form with at most 3 literals in each clause $c \in C$. \\
{\sc question}: Is there a 0-1 assignment for $X$ so that for all 0-1 assignments of $Y$ there is a 0-1 assignment of $Z$ such that $E$ is satisfied?
}
\vspace{0.3cm}

\vspace{0.3cm}
\boxxx{
\textsc{\textbf{Trilevel Interdiction Knapsack (TIK)}}: \\
{\sc instance}: A set of items $O$ such that each $o \in O$ has two  positive integer weights $a'_o$ and $a_o$ and a positive integer profit $p_o$, two positive integer maximum weight capacities $A'$ and $A$, a positive integer maximum profit $B$ and a positive integer goal $\bar{K} \leq B$.\\
{\sc question}: Is there a subset $O_1\subseteq O$ of items, with $\sum_{o \in O_1} a'_o \leq A'$, such that every subset $O_2\subseteq O\setminus O_1$, with $\sum_{o \in O_2} a_o \leq A,$ there is a subset $O_3 \subseteq O \setminus O_2$, with  $\sum_{o \in O_3} p_o \leq B$,  such that $\sum_{o \in O_3} p_o \geq \bar{K}$ holds?
}
\vspace{0.3cm}

\begin{thm}
TIK is $\Sigma_3^p$-complete.
\end{thm}
\begin{proof}
The statement of TIK is of the form $\exists O_1 \ \ \forall O_2 \ \  \exists O_3 \ \ Q(O_1,O_2,O_3)$, directly implying that it is in $\Sigma_3^p$.

Next, we use a reduction from the $B_3 \cap 3CNF$ which is very much in line with the reduction from 3-SAT to Subset Sum presented in~\cite[Theorem 34.15]{IntrAlg2009}:
\begin{itemize}
    \item  For each variable $u \in U$, we create two items $o_u$ and $o_{\bar{u}}$, one for each possible 0-1 assignment of $u$. We designate by $O_U = \lbrace o_u: u \in U \rbrace$ and $O_{\bar{U}} = \lbrace o_{\bar{u}}: u \in U \rbrace$ the two sets of items of size $\vert U \vert$. 
    \item For each clause $c \in C$, \emph{(i)}  if $c$ has 1 literal, we create one item $o^1_c$, \emph{(ii)} if $c$ has 2 literals, we create two items $o^1_c$ and $o^2_c$, and \emph{(iii)} if $c$ has 3 literals, we create three items $o^1_c$, $o^2_c$ and $o^3_c$. We designate by $O_C$ the set of items associated with $C$.
    \item Weights, profits, maximum capacities, maximum profit and goal will be given by digits of size $\vert X \vert +\vert Y \vert +\vert Z \vert +\vert C\vert +1$ in base 10. Hence, each digit position is labeled by a variable or a clause: the first $\vert C \vert$ positions (least significant numbers) are labeled by the clauses, then the next $\vert X \vert$ positions are labeled by the variables $X$, then the next $\vert Y \vert$ positions are labeled  by the variables $Y$, then the next $\vert Z \vert$ positions are labeled by the variables $Z$, and, finally, the last position is labeled as \emph{forbidden}.
        \begin{itemize}
            \item  For each  $u \in U$, the two corresponding items $o_u$ and $o_{\bar{u}}$ have weights and profits as described next.     The weights and profits $a'_{o_u}$, $a_{o_u}$, $p_{o_u}$, $a'_{o_{\bar{u}}}$,  $a_{o_{\bar{u}}}$ and $p_{o_{\bar{u}}}$ have digit 1 in the position labeled by the variable $U$ and 0 in the positions labeled by other variables; the remaining digits are zero for $a'_{o_u}$, $a_{o_u}$, $a'_{o_{\bar{u}}}$ and  $a_{o_{\bar{u}}}$. In particular, for all $o \in O_{U} \cup O_{\bar{U}}$, it holds $a'_{o_u}=a_{o_u}$ and $a'_{o_{\bar{u}}}=a_{o_{\bar{u}}}$.
            
            If the literal $u$ appears in clause $c \in C$, then $p_{o_u}$ has digit 1 in the position labeled as $c$, and 0 otherwise. Similarly, if the literal $\neg u$ appears in clause $c \in C$,  $p_{o_{\bar{u}}}$ has digit 1 in the position labeled by $c$, and 0 otherwise. Finally, for all $o \in O_{U} \cup O_{\bar{U}}$, $p_{o_u}$ and $p_{o_{\bar{u}}}$  have digit 0 in the position labeled as forbidden. 
            \item For each $c \in C$, the associated items have weights and profits as follows. If $c$ has one literal, $a'_{o^1_c}$ and $a_{o^1_c}$ have 1 in the position labeled as forbidden and 0 elsewhere; $p_{o^1_c}$ has digit 3 in the position labeled as $c$ and 0 elsewhere.  If $c$ has two literals, $a'_{o^1_c}$, $a'_{o^2_c}$, $a_{o^1_c}$ and $a_{o^2_c}$ have 1 in the position labeled as forbidden and 0 elsewhere; $p_{o^1_c}$ and $p_{o^2_c}$ have digit 3 and 2, respectively, in the position labeled as $c$ and 0 elsewhere.    If $c$ has three literals, $a'_{o^1_c}$, $a'_{o^2_c}$,  $a'_{o^3_c}$, $a_{o^1_c}$,  $a_{o^2_c}$ and $a_{o^3_c}$ have 1 in the position labeled as forbidden and 0 elsewhere; $p_{o^1_c}$,  $p_{o^2_c}$ and $p_{o^3_c}$ have digit 3, 2 and 1, respectively, in the position labeled as $c$ and 0 elsewhere.  
            \item The weight capacity $A'$ has 1s for all digits with labels in $X$ and 0s elsewhere. Hence, $O_1$ cannot contain items from $\lbrace o_u, o_{\bar{u}}: u \in Z \cup Y \rbrace \cup O_C$.
            \item The weight capacity $A$  has 1s for all digits with labels in $Y$, 2s for all digits with labels in $X$  and 0s elsewhere. Hence, $O_2$ cannot contain items from $\lbrace o_u, o_{\bar{u}}: u \in Z \rbrace \cup O_C$.
            \item The maximum profit $B$ has 1s for all digits with labels in $X \cup Z$, 2s for all digits with labels in $Y$, 4s for all digits with labels in $C$, and 0s elsewhere. Hence, $O_3$ can take any item (as long as not interdicted by $O_2$).
            \item We make $\bar{K}$ is equal to $B$, except for the digits with labels $Y$, where it is 1. 
        \end{itemize}
\end{itemize}

See Figure~\ref{B3_tik} for an illustration of our reduction. 

\begin{figure}
    \centering
    \resizebox{0.58\textwidth}{!}{
    \begin{tabular}{ll|c|c:c:cc|ccc}
     $O$   &    &        & $Z$ & $Y$ & \multicolumn{2}{c|}{$X$} & \multicolumn{3}{c}{$C$} \\  
             &       & forbidden & $d$ & $c$ & $b$ & $a$ & $c_1$ & $c_2$ & $c_3$ \\ \midrule 
    $o_a$ & $a'_{o_a}=a_{o_a}$           & 0         & 0   & 0   & 0   & 1   &  0    &  0    & 0 \\
       & $p_{o_a}$           & 0         & 0   & 0   & 0   & 1   &  1    &  0    & 1\\
   $o_{\bar{a}}$ &$a'_{o_{\bar{a}}}=a_{o_{\bar{a}}}$  & 0         & 0   & 0   & 0   & 1   &  0    &  0    & 0\\
   &$p_{o_{\bar{a}}}$  & 0         & 0   & 0   & 0   & 1   &  0    &  1    & 1\\
   $o_b$  &$a'_{o_b}=a_{o_b}$          & 0         & 0   & 0   & 1   & 0   &  0    &  0    & 0\\
     &$p_{o_b}$          & 0         & 0   & 0   & 1   & 0   &  1    &  0    & 1\\
    $o_{\bar{b}}$&$a'_{o_{\bar{b}}}=a_{o_{\bar{b}}}$ & 0         & 0   & 0   & 1   & 0   &  0    &  0    & 0\\
    &$p_{o_{\bar{b}}}$ & 0         & 0   & 0   & 1   & 0   &  0    &  1    & 0\\
$o_c$    &$a_{o_c}=a'_{o_c}$           & 0         & 0   & 1   & 0   & 0   &  0    &  0    & 0\\ 
 &$p_{o_c}$           & 0         & 0   & 1   & 0   & 0   &  0    &  0    & 1\\ 
$o_{\bar{c}}$    &$a_{o_{\bar{c}}}=a'_{o_{\bar{c}}}$ & 0         & 0   & 1   & 0   & 0   &  0    &  0    & 0\\
 &$p_{o_{\bar{c}}}$ & 0         & 0   & 1   & 0   & 0   &  1    &  0    & 1\\
$o_d$    &$a'_{o_d}=a_{o_d}$         & 0         & 1   & 0   & 0   & 0   &  0    &  0    & 0\\
 &$p_{o_d}$         & 0         & 1   & 0   & 0   & 0   &  0    &  1    & 0\\
$o_{\bar{d}}$    &$a'_{o_{\bar{d}}}=a_{o_{\bar{d}}}$  & 0         & 1   & 0   & 0   & 0   &  0    &  0    & 0\\
 &$p_{o_{\bar{d}}}$  & 0         & 1   & 0   & 0   & 0   &  0    &  0    & 0\\
    $o^1_{c_1}$ & $a'_{o^1_{c_1}}$  &  1 & 0 & 0 & 0 & 0 & 0 & 0 & 0\\
                & $a_{o^1_{c_1}}$  &  1 & 0 & 0 & 0 & 0 & 0 & 0 & 0\\
                & $p_{o^1_{c_1}}$  &  0 & 0 & 0 & 0 & 0 & 3 & 0 & 0\\
    $o^2_{c_1}$ & $a'_{o^2_{c_1}}$  &  1 & 0 & 0 & 0 & 0 & 0 & 0 & 0\\
                & $a_{o^2_{c_1}}$  &  1 & 0 & 0 & 0 & 0 & 0 & 0 & 0\\
                & $p_{o^2_{c_1}}$  &  0 & 0 & 0 & 0 & 0 & 2 & 0 & 0\\
    $o^3_{c_1}$ & $a'_{o^3_{c_1}}$  &  1 & 0 & 0 & 0 & 0 & 0 & 0 & 0\\
                & $a_{o^3_{c_1}}$  &  1 & 0 & 0 & 0 & 0 & 0 & 0 & 0\\
                & $p_{o^3_{c_1}}$  &  0 & 0 & 0 & 0 & 0 & 1 & 0 & 0\\
    $o^1_{c_2}$ & $a'_{o^1_{c_2}}$  &  1 & 0 & 0 & 0 & 0 & 0 & 0 & 0\\
                & $a_{o^1_{c_2}}$  &  1 & 0 & 0 & 0 & 0 & 0 & 0 & 0\\
                & $p_{o^1_{c_2}}$  &  0 & 0 & 0 & 0 & 0 & 0 & 3 & 0\\
    $o^2_{c_2}$ & $a'_{o^2_{c_2}}$  &  1 & 0 & 0 & 0 & 0 & 0 & 0 & 0\\
                & $a_{o^2_{c_2}}$  &  1 & 0 & 0 & 0 & 0 & 0 & 0 & 0\\
                & $p_{o^2_{c_2}}$  &  0 & 0 & 0 & 0 & 0 & 0 & 2 & 0\\
    $o^3_{c_2}$ & $a'_{o^3_{c_2}}$  &  1 & 0 & 0 & 0 & 0 & 0 & 0 & 0\\
                & $a_{o^3_{c_2}}$  &  1 & 0 & 0 & 0 & 0 & 0 & 0 & 0\\
                & $p_{o^3_{c_2}}$  &  0 & 0 & 0 & 0 & 0 & 0 & 1 & 0\\
    $o^1_{c_3}$ & $a'_{o^1_{c_3}}$  &  1 & 0 & 0 & 0 & 0 & 0 & 0 & 0\\
                & $a_{o^1_{c_3}}$  &  1 & 0 & 0 & 0 & 0 & 0 & 0 & 0\\
                & $p_{o^1_{c_3}}$  &  0 & 0 & 0 & 0 & 0 & 0 & 0 & 3\\
    $o^2_{c_3}$ & $a'_{o^2_{c_3}}$  &  1 & 0 & 0 & 0 & 0 & 0 & 0 & 0\\
                & $a_{o^2_{c_3}}$  &  1 & 0 & 0 & 0 & 0 & 0 & 0 & 0\\
                & $p_{o^2_{c_3}}$  &  0 & 0 & 0 & 0 & 0 & 0 & 0 & 2\\
    $o^3_{c_3}$ & $a'_{o^3_{c_3}}$  &  1 & 0 & 0 & 0 & 0 & 0 & 0 & 0\\
                & $a_{o^3_{c_3}}$  &  1 & 0 & 0 & 0 & 0 & 0 & 0 & 0\\
                & $p_{o^3_{c_3}}$  &  0 & 0 & 0 & 0 & 0 & 0 & 0 & 1 \\ \midrule
    &$A'$ & 0 & 0 & 0 & 1 & 1 & 0 & 0 & 0 \\
    &$A$ & 0 & 0 &1 & 2 & 2 & 0 & 0 & 0  \\
    &$B$ & 0 & 1 & 2 & 1 & 1 & 4 & 4 & 4  \\
    &$\bar{K}$  & 0 & 1 & 1 & 1 & 1 & 4 & 4 & 4  
    \end{tabular}
    }
    \caption{Example of construction of TIK from an instance $B_3 \cap 3CNF$ with $E =(a \lor b \lor \neg c) \land (\neg a \lor \neg b \lor d) \land (a \lor c \lor b) $, where $X = \{a, b \}$, $Y = \{c\}$, $Z=\{d\}$ and the clauses are labeled from left to right.}
    \label{B3_tik}
\end{figure}

Let $B_3 \cap 3CNF$ be a \emph{Yes} instance. Then, take in $O_1$ the items $o_u$ such that $u \in X$ is 1 and the items $o_{\bar{u}}$, otherwise. Clearly, this choice of $O_1$ respects the maximum weight $A'$. By construction, given this $O_1$, the best $O_2$ will take all items associated with $X$ and not taken by $O_1$, as it does not interfere with the budget left for the items associated with $Y$. Furthermore, the optimal $O_2$ will also take exactly one of the items $o_u$ or $o_{\bar{u}}$ for $u \in Y$: 
\begin{itemize}
    \item The two items associated with the most significant digit whose label is in $Y$ cannot be taken simultaneously in $O_2$ as it would violate the weight capacity $A$. In fact, exactly one of these items must be taken, as otherwise $O_3$ would select them both, making the achievement of the profit $\bar{K}$ only dependent on the items associated with the $Z$; consequently, the goal would be achieved.
    \item The two items associated with the second most significant digit whose label is in $Y$ cannot be taken simultaneously, since we already know that one of the items associated with the most significant digit in $Y$ is taken which would result in a violation of the weight capacity $A$. Hence, reasoning as before, $O_2$ will take exactly of the items associated with the second most significant digit in $Y$.
    \item The reasoning above propagates until the least significant digit labeled in $Y$. We conclude that the best $O_2$ will have exactly one of the items $o_u$ or $o_{\bar{u}}$ for $u \in Y$.
\end{itemize}
Finally, $O_3$ will contain $O_1$ and all the items associated with $Y$ not in $O_2$. This makes the rest of the items selection for $O_3$ completely equivalent to variable assignment in $Z$ for $B_3 \cap CNF$ (precisely, the standard reduction from 3-SAT to Subset Sum). Therefore, $TIK$ is a \emph{Yes} instance.

Next, suppose that  TIK is a \emph{Yes} instance. Certainly, an optimal $O_1$ must have exactly one of the items $o_u$ and $o_{\bar{u}}$ for $u \in X$, otherwise, $O_2$ could interdict some   $o_u$ and $o_{\bar{u}}$, making the goal $\bar{K}$ impossible to be achieved. As argued before, an optimal reaction $O_2$ to $O_1$ will select the items associated with $X$ not in $O_1$. 

Assign 1 to $u \in X$ such that $o_u \in O_1$, and 0 otherwise. For any valid assignment of the variables in $Y$, the correspondence in TIK is the following: if $u \in Y$ is 1, add $o_{\bar{u}}$ to $O_2$, otherwise add $o_u$. This forces $O_3$ to select for each $u \in Y$, $o_u$ if $u$ is 1 and $o_{\bar{u}}$ if $u$ is 0; otherwise, the goal $\bar{K}$ is not attained. Since, by hypothesis, TIK is a \emph{Yes} instance, for those $O_1$ and $O_2$, there is $O_3$ such that the profit $\bar{K}$ is exactly achieved which implies that there is an assignment of $Z$ such that $E$ is satisfied.

\end{proof}

\begin{thm}
MCN$_w$ is $\Sigma_3^p$-complete, even on trees.
\label{th_MCN_w}
\end{thm}
\begin{proof}
MCN$_w$ is clearly in $\Sigma_3^p$.

Next, from an instance of TIK, we construct the following instance of MCN$_w$:
\begin{itemize}
    \item Let $\Omega=A'$, $\Phi=A+1$, $\Lambda = B$ and $K=\bar{K}$.
    \item  For each item $o \in O$ create three vertices $v^1_o$, $v_o^2$ and $v_o^3$ with 
    \begin{itemize}
        \item $\hat{c}_{v^1_o}=\Omega+1$, $h_{v^1_o}=\Phi+1$, $c_{v^1_o}=p_o$ and $b_{v^1_o}=0$; this vertex is only available for the protection set $P$;
        \item $\hat{c}_{v^2_o}=\Omega+1$, $h_{v^2_o}=\Phi+1$, $c_{v^2_o}=\Lambda+1$ and $b_{v^2_o}=p_o$; this vertex cannot be vaccinated, directly infected or protected;
        \item $\hat{c}_{v^3_o}=a'_o$, $h_{v^3_o}=a_o$, $c_{v^3_o}=\Lambda+1$ and $b_{v^3_o}=0$; this vertex is only available for the vaccination set $D$ and for the direct infection set $I$;
    \end{itemize}
    \item Create a vertex $r$ with   $\hat{c}_{r}=\Phi+1$, $h_{r}=1$, $c_{r}=1$ and $b_{r}= K$.
    \item For each item $o \in O$, add the edges $(r, v^1_o)$, $(v^1_o,v^2_o)$  and $(v^2_o,v^3_o)$.
\end{itemize}
See Figure~\ref{fig:TIK_MCNw} for an illustration of our reduction.

\begin{figure}[h]
    \resizebox{\columnwidth}{!}{
    \input{Paper_Version/FigureTIKMCN}}
    \caption{Graph reduction from TIK to MCN$_w$ when $O=\lbrace 1,2 \ldots, n \rbrace$. The only vertices resulting in positive benefit are the ones in white. The vertices in gray can be vaccinated and directly attacked. The vertices in green can be protected. The vertex in black can be attacked (and protected).}
    \label{fig:TIK_MCNw}
\end{figure}
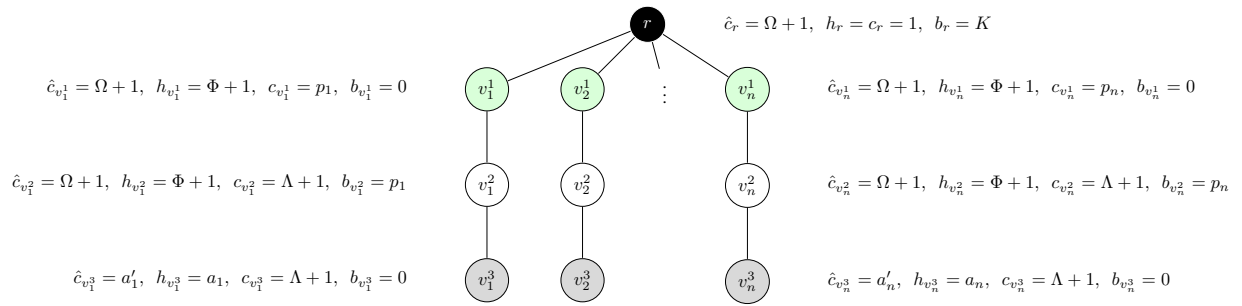

The key ingredients of this reduction are the following: \emph{(i)} independently of the vaccination strategy, an optimal attack will always include the vertex $r$, \emph{(ii)} hence, the only way to collect a positive benefit $p_o$ is by ensuring that vertex $v^2_o$ is saved, \emph{(iii)} the latter is only possible if $v^3_o$ is vaccinated and $v^1_o$ is protected or if $v^3_o$ is not attacked and $v^1_o$ is protected. These observations allow to show that TIK is a \emph{Yes} instance if and only if MCN$_w$ is a \emph{Yes} instance. The remainder of the proof follows a similar reasoning to the previous proofs for the weighted games. 


\end{proof}

%% file: Paper_Version/Figure_weight_2.tex
\begin{tikzpicture}[-,>=stealth',shorten >=0.2pt,auto,node distance=1cm,baseline={(current bounding box.north)},  scale=0.8, transform shape,
  main node/.style={circle,fill=black!15,draw,minimum size=0.7cm},spe node/.style={fill=white,draw=white}]

  \node[main node] (6) {$r$};
  \node[main node] (2) [above right=of 6] {$v_1$};
  \node[spe node] (7) [left =of 6] {$b_r=\sum_{o=1}^np_{o}+1, \ \ h_r=c_r=1$};
  \node[spe node] (8) [right=of 2] {$b_{v_1}=c_{v_1}=p_1,\ \ h_{v_1}=a_1$};
  \node[main node] (3) [right =of 6] {$v_2$};
  \node[spe node] (9) [right=of 3] {$b_{v_2}=c_{v_2}=p_2,\ \ h_{v_2}=a_2$};
  \node[spe node] (4) [below right=of 6] {$\ldots$};
  \node[main node] (5) [below=of 6] {$v_n$};
  \node[spe node] (10) [right=of 5] {$b_{v_n}=c_{v_n}=p_{n}, \ \ h_{v_n}=a_n$};

  \path[every node/.style={font=\sffamily\small}]
    (6) 
        edge node {}  (2)
        edge node {}  (3)
        edge node {}  (4)
        edge node {}  (5);
\end{tikzpicture}

%% file: Paper_Version/Figure_weight_1.tex
\begin{tikzpicture}[-,>=stealth',shorten >=0.2pt,auto,node distance=1cm,baseline={(current bounding box.north)},  scale=0.8, transform shape,
  main node/.style={circle,fill=black!15,draw,minimum size=0.7cm},spe node/.style={fill=white,draw=white}]

  \node[main node] (6) {$r$};
  \node[main node] (2) [above right=of 6] {$v_1$};
   \node[spe node] (7) [left =of 6] {$b_r=\bar{K},\ \ \hat{c}_r=h_r=1$};
  \node[spe node] (8) [right=of 2] {$b_{v_1}=h_{v_1}=p_1,\ \ \hat{c}_{v_1}=a_1$};
  \node[main node] (3) [right =of 6] {$v_2$};
  \node[spe node] (9) [right=of 3] {$b_{v_2}=h_{v_2}=p_2,\ \ \hat{c}_{v_2}=a_2$};
  \node[spe node] (4) [below right=of 6] {$\ldots$};
  \node[main node] (5) [below=of 6] {$v_n$};
  \node[spe node] (10) [right=of 5] {$b_{v_n}=h_{v_n}=p_n, \ \ \hat{c}_{v_n}=a_n$};

  \path[every node/.style={font=\sffamily\small}]
    (6) 
        edge node {}  (2)
        edge node {}  (3)
        edge node {}  (4)
        edge node {}  (5);
\end{tikzpicture}

%% file: Paper_Version/FigureTIKMCN.tex
\begin{tikzpicture}[-,>=stealth',shorten >=0.2pt,auto,node distance=1cm,baseline={(current bounding box.north)},  scale=0.8, transform shape,
inf node/.style={circle,fill=black,draw=white,minimum size=0.7cm},
pro node/.style={circle,fill=green!15,draw,minimum size=0.7cm},
ben node/.style={circle,draw,minimum size=0.7cm},
  main node/.style={circle,fill=black!15,draw,minimum size=0.7cm},spe node/.style={fill=white,draw=white}]

  \node[inf node] (10) {\textcolor{white}{$r$}};
  \node[pro node] (2) [below left=of 10] {$v^1_2$};
  \node[pro node] (3) [left=of 2] {$v^1_1$};
  \node[spe node] (4) [right=of 2] {$\vdots$};
  \node[ben node] (5) [below=of 3] {$v^2_1$};
  \node[main node] (6) [below=of 5] {$v^3_1$};
  \node[ben node] (7) [below=of 2] {$v^2_2$};
  \node[main node] (8) [below=of 7] {$v^3_2$};
  \node[pro node] (11) [right=of 4] {$v^1_n$};
  \node[ben node] (9) [below=of 11] {$v^2_n$};
  \node[main node] (12) [below=of 9] {$v^3_n$};
  \node[spe node] (13) [right=of 10] {$\hat{c}_r= \Omega+1, \ \ h_r=c_r=1, \ \ b_r=K$};
  \node[spe node] (14) [left=of 3] {$\hat{c}_{v^1_1}= \Omega+1, \ \ h_{v^1_1}=\Phi +1, \ \ c_{v^1_1}=p_1, \ \ b_{v^1_1}=0$};
  \node[spe node] (15) [left=of 5] {$\hat{c}_{v^2_1}= \Omega+1, \ \ h_{v^2_1}=\Phi +1, \ \ c_{v^2_1}=\Lambda+1, \ \ b_{v^2_1}=p_1$};
  \node[spe node] (16) [left=of 6] {$\hat{c}_{v^3_1}=a'_1, \ \ h_{v^3_1}=a_1, \ \ c_{v^3_1}=\Lambda+1, \ \ b_{v^3_1}=0$};
  
   \node[spe node] (17) [right=of 11] {$\hat{c}_{v^1_n}= \Omega+1, \ \ h_{v^1_n}=\Phi +1, \ \ c_{v^1_n}=p_n, \ \ b_{v^1_n}=0$};
  \node[spe node] (18) [right=of 9] {$\hat{c}_{v^2_n}= \Omega+1, \ \ h_{v^2_n}=\Phi +1, \ \ c_{v^2_n}=\Lambda+1, \ \ b_{v^2_n}=p_n$};
  \node[spe node] (19) [right=of 12] {$\hat{c}_{v^3_n}=a'_n, \ \ h_{v^3_n}=a_n, \ \ c_{v^3_n}=\Lambda+1, \ \ b_{v^3_n}=0$};

  \path[every node/.style={font=\sffamily\small}]
    (10) edge node {}  (2)
         edge node {}  (3)
         edge node {}  (11)
         edge node {}  (4)
    (2)  edge node {}  (7)
    (7)  edge node {}  (8)
    (3)  edge node {}  (5)
    (5)  edge node {}  (6)
    (11) edge node {}  (9)
    (9)  edge node {}  (12);
\end{tikzpicture}

%% file: Paper_Version/Directed_Graphs.tex
\section{Directed graphs}\label{sec:directed_uni}
In this section, we consider directed graphs $G=(V,A)$ and restrict costs and benefits to be unitary. We use the subscript $_{dir}$ for these problem versions.  Clearly, these problems inherit the complexity of their unitary undirected versions, as they are more general. In fact, we were able to go a level up in the polynomial hierarchy  for some of its subgames in comparison with the unitary undirected cases. In this section, we first prove that the {\sc Attack}$_{dir}$ is NP-complete, and then demonstrate that {\sc Vaccination-Attack}$_{dir}$ is $\Sigma_2^p$-complete. Later, in Section~\ref{sec:special}, we present special properties of {\sc Protect}$_{dir}$ that allow us to easily prove NP-com\-plete\-ness for directed acyclic graphs and polynomiality for arborescences.

It should be remarked that we do not address {\sc Attack-Protect}$_{dir}$ and thus, it remains open whether it is $\Sigma_2^p$-complete. The difficulty on dealing with this subgame is related to the lack of $\Sigma_2^p$-hard problems involving unitary parameters or a division on the two players decision variables: in  {\sc Attack-Protect}$_{dir}$ all parameters are 1 and all vertices can be subject to infection or protection. On the other hand, as an example, non-trivial instances of KIP (presented in Section~\ref{subsec:undirected_wei_AP}) should have weights not all 1, otherwise it becomes polynomially solvable as it can be reduced to its continuous version and, consequently, efficiently solved~\cite{CarvalhoMarcotteAndrea}. Another example, {\sc 2-CNF-Alternating Quantified Satisfiability}, to be introduced in Section~\ref{subsec:vac_attack_dir}, and which is  $\Sigma_2^p$-complete, demands each player to control distinct sets of variables. For {\sc Vaccination-Attack}$_{dir}$, we were able to bypass this challenge but an analogous trick does not seem easily adaptable for {\sc Attack-Protect}$_{dir}$.


\subsection{The {\sc Attack}$_{dir}$ problem}

First, we study the \textit{Attack problem} on directed graphs, {\sc Attack}$_{dir}$. We are given a directed graph $G_a$ resulting from the deletion of the vaccinated vertices from the original graph, and an integer budget $\Phi$. In this setting, there is no protection phase, \ie ~$\Lambda = 0$. The decision version of the problem is: 

\vspace{0.3cm}
\boxxx{
\textbf{\textsc{Attack$_{dir}$}}: \\
{\sc instance}: A directed graph $G_a = (V_a, A_a)$, a non-negative integer budget $\Phi \leq |V_a|$, and a non-negative integer $K$. \\
{\sc question}: Is there a subset of vertices $I \subseteq V_a$, $|I| \leq \Phi$ such that the number of infected vertices in $G_a$ is greater or equal to $K$?
}
\vspace{0.3cm}

We saw that in the undirected case, this problem is solvable in linear time, the best strategy being to infect the $\Phi$ largest connected components of $G_a$. But in the directed case, the infection is only allowed to propagate itself according to the direction of the arcs, which makes the problem of choosing the right set of vertices to attack NP-complete. We will use a reduction from the \textit{$3$-Satisfiability problem}, which is one of the Karp's $21$ NP-complete problems \cite{Karp1972}.

\vspace{0.3cm}
\boxxx{
\textbf{\textsc{3-satisfiability (3-SAT):}} \\
{\sc instance}: Set $U$ of variables, Boolean expression $E$ over $U$ in conjunctive normal form with exactly $3$ literals in each clause $c \in C$. \\
{\sc question}: Is there a $0$-$1$ assignment for the variables in $U$ that satisfies $E$?
}
\vspace{0.3cm}

\begin{thm}
\textsc{Attack$_{dir}$} is NP-complete, even on directed acyclic graphs.
\label{A_dir_NP_hard}
\end{thm}

\begin{proof}
\textsc{Attack$_{dir}$} $\in$ NP as, given a set of attacked vertices $I$, checking whether the set of infected vertices is greater than $K$ is easily done using a DFS. \\
To prove that \textsc{Attack$_{dir}$} is NP-hard, we take an instance of $3$-SAT. We build a directed acyclic graph $G_a$ as follows:
\begin{itemize}
    \item For each variable $u \in U$, we create two vertices $v_u$ and $v_{\bar{u}}$, one for each possible $0$-$1$ assignment of $u$. We call $V_U = \{v_u; u \in U\}$ and $V_{\bar{U}}=\{v_{\bar{u}}; u \in U\}$ the two sets of vertices of size $|U|$. For each variable $u$, we also create a directed path $p_u$ of length $|C|+|U|-1$, with an in-going arc from both $v_u$ and $v_{\bar{u}}$ at the beginning of the path.
    \item For each clause $c \in C$, we create a vertex $v_c \in V_C$.
    \item From each vertex $v_u \in V_U$, we draw an arc $(v_u, v_c)$ to every clause in which the positive literal $u$ appears. Similarly, we draw an arc $(v_{\bar{u}},v_c)$ from each $v_{\bar{u}} \in V_{\bar{U}}$ to every clause in which the negative literal $\neg u$ appears.
\end{itemize}
An example of this construction can be found in \hyperref[Fig2]{Figure \ref{Fig2}}. We set $\Phi = |U|$, $K = |U| \times (|U| + |C|) + |C|$ and argue that answering \textsc{Attack$_{dir}$} on this instance is the same as answering $3$-SAT.

Indeed,  suppose that $3$-SAT is a \emph{Yes} instance, \ie ~there is a 0-1 assignment to the variables in $U$ such that every clause in $E$ is true. Taking this assignment, by attacking $v_u$ if $u$ is set to be $1$ and $v_{\bar{u}}$ otherwise, we attack exactly $\Phi$ vertices in $G_a$. Moreover, each path $p_u$ is infected, and for each pair $(v_u, v_{\bar{u}})$, there is exactly one vertex infected due to the direction of the arcs. Finally, as $E$ is true, each clause $c$ is true, which translates into the fact that each $v_c$ in the graph $G_a$ is infected. Overall, there are exactly $|U| + |U| \times |p_u| + |C| = |U| \times (|U| + |C|) + |C|$ vertices infected in the graph.

Conversely, we prove that if {\sc Attack}$_{dir}$ is a \emph{Yes} instance, \ie, there is a feasible attack $I^*$ on $G_a$ leading to at least $K=|U| \times (|U| + |C|) + |C|$ vertices infected, then  $E$ is satisfiable and the corresponding 0-1 assignment can be read in $I$. Let $I^*$ be such an attack strategy. First, we remark that the largest possible set of infected vertices should contain all the vertices $V_{p_u}$ of each path $p_u$: it is possible to infect them all as $\Phi = |U|$ and due to their size equal to $|C| +|U|-1$, we can prove that not infecting all of them results in a sub-optimal solution. Indeed, suppose that for one $u'$ we do not infect any of the vertices $V_{p_{u'}}$ of the path $p_{u'}$. Let $\alpha^*$ be the maximum number of vertices we can infect without infecting $p_{u'}$. As $p_{u'}$ is not infected, $v_{u'}$ and $v_{\bar{u}'}$ cannot be either. Thus, an easy upper bound $\alpha_{up}$ on $\alpha^*$ is obtained by saying that every vertex of the graph is infected, except for the ones in $\{v_{u'}, v_{\bar{u}'}\} \cup V_{p_{u'}}$. Then,
\begin{align*}
    \alpha^* \leq \alpha_{up} &= (|U|-1)\times|p_u| + 2(|U|-1) + |C| \\
    &= (|U| - 1 )\times(|U|+|C| -1) + 2|U| -2 +|C|\\
    & = |U|^2 + |U| \times |C| - 2|U| -|C| +1 + 2|U| - 2 + |C| \\
    & = |U| \times (|U| + |C|) -1.
\end{align*}
As we assumed that the optimal attack $I^*$ infected at least $K=|U| \times (|U| + |C|) + |C|$ vertices, which is strictly greater than $\alpha_{up}$, we proved that no strategy not infecting all the paths can infect $K$ vertices. \\
Thus, as there is exactly $\Phi$ different paths, we should attack exactly one element in each set of vertices $\{v_u, v_{\bar{u}}\} \cup V_{p_u}$: if we attacked more than one, then the remaining budget would not allow to attack all the paths.  As attacking $v_u$ or $v_{\bar{u}}$ leads to a strictly greater number of infected vertices than infecting a vertex in $p_u$, there is no harm in assuming that no vertex inside the $p_u$ is in $I^*$. This implies that $I^* \subset V_U \cup V_{\bar{U}}$. At this point, there are at least $|p_u| \times |U| + |U| = |U| \times (|U| + |C|)$ vertices infected. Since we supposed that we had a \textit{Yes} instance to \textsc{Attack$_{dir}$}, there must be $K = |U| \times (|U| + |C|) + |C|$ infected vertices, which implies that all vertices in $V_C$ are infected. Thus, $3$-SAT is a \textit{Yes} instance and $I^*$ is a 0-1 assignment of $U$ that makes $E$ true, concluding the proof.
\end{proof}

\begin{rem}
Note that  the proof of Theorem~\ref{A_dir_NP_hard} holds if $p_u$ is replaced by  a complete graph with  $|C|+|U|-1$ vertices (the length of the path). This observation will be useful for the reduction used in {\sc Vaccination-Attack}$_{dir}$
\end{rem}

\begin{figure}[ht]
    \centering
    \input{Paper_Version/Figure2}
    \caption{Example of construction of $G_a$ from the boolean expression in CNF with 3 literals in each clause $E = (a \lor b \lor \neg c) \land (\neg a \lor b \lor c) \land (a \lor \neg b \lor c) $. We have $U = \{a,b,c \}$ and $|C| = 3$. Taking $I = \{v_a, v_b, v_c \}$ is optimal.}
    \label{Fig2}
\end{figure}
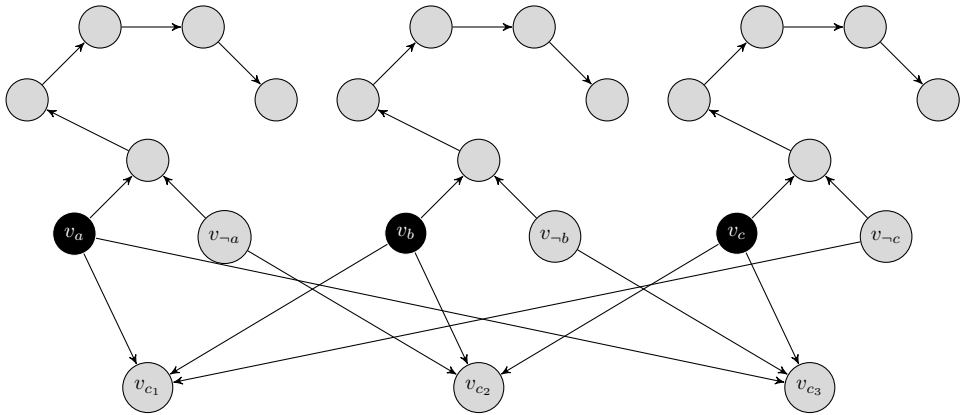

\subsection{The {\sc Vaccination-Attack}$_{dir}$ problem}\label{subsec:vac_attack_dir}

Our demonstration of NP-com\-plete\-ness for \textsc{Attack$_{dir}$} inspires our proof for the  $\Sigma_2^p$-completeness of {\sc Vaccination-Attack}$_{dir}$. The formulation of this decision problem is 

\vspace{0.3cm}
\boxxx{
\textsc{\textbf{Vaccination-Attack}$_{dir}$}: \\
{\sc instance}: A graph $G=(V,A)$, two non-negative integer budgets $\Omega$ and $\Phi$ such that $\Omega + \Phi \leq |V|$ and a non-negative integer $K$.
\\
{\sc question}: Is there a subset $D \subseteq V$, $|D| \leq \Omega$ such that $\forall I \subseteq V\backslash D$ with $|I| \leq \Phi$, the number of infected vertices $|V| - s(G, D, I, \emptyset)\leq K$?
}
\vspace{0.3cm}


We will use a reduction from a variant of the \textit{2-Alternating Quantified Satisfiability problem} ($B_2$). Historically, $B_2$ was the first problem shown to be $\Sigma_2^p$-complete \cite{MeyerStock72}. If the Boolean formula studied in $B_2$ is in DNF with $3$ literals per clause, then the problem is still $\Sigma_2^p$-complete \cite{Wrathall1976}. Thus, if we consider expressions in CNF with $3$ literals per clause, instead of seeking to \textit{satisfy} the Boolean formula, we should state the question as formulated in \cite{Johannes2011NewCO}:

\vspace{0.3cm}
\boxxx{
\textbf{\textsc{2-CNF-Alternating Quantified Satisfiability}} ($B_2^{CNF}$):\\
{\sc instance}: Disjoint non-empty sets of variables $X$ and $Y$, Boolean expression $E$ over $U = X \cup Y$ in conjunctive normal form with exactly 3 literals in each clause. \\
{\sc question}: Is there a 0-1 assignment for $X$ so that there is no 0-1 assignment for $Y$
such that $E$ is satisfied?
}
\vspace{0.3cm}

\begin{thm}
{\sc Vaccination-Attack}$_{dir}$ is $\Sigma_2^p$-complete.
\label{th_VA_dir}
\end{thm}

\begin{proof}
From the formulation in the form of $\exists D \ \ \forall I \ \ Q(D,I)$, we deduce that {\sc Vaccination-Attack}$_{dir}$ $\in \Sigma_2^p$. \\
To show that it is $\Sigma_2^p$-hard, we take an instance of $B_2^{CNF}$. We build $G$ in a similar fashion to how $G_a$ was built in the proof of the \hyperref[A_dir_NP_hard]{Theorem \ref{A_dir_NP_hard}}, the main difference being the use of cliques instead of paths. However, to differentiate the variables in $X$ from the ones in $Y$, we slightly change the construction:
\begin{itemize}
    \item For each variable $x \in X$, we create two vertices $v_x$ and $v_{\bar{x}}$, one for each possible $0$-$1$ assignment of $x$. We call $V_X$ and $V_{\bar{X}}$ the sets of $v_x$ and $v_{\bar{x}}$. We also create two cliques $k_x$ and $k_{\bar{x}}$ of $|C|+|Y|-1$ vertices $V_{k_x}$ and $V_{k_{\bar{x}}}$. 
    \item For each variable $y \in Y$, we create two vertices $v_y$ and $v_{\bar{y}}$, one for each possible $0$-$1$ assignment of $y$. Let $V_Y$ and $V_{\bar{Y}}$ be these two sets of vertices, and $V_U = V_X \cup V_Y$, $V_{\bar{U}} = V_{\bar{X}} \cup V_{\bar{Y}}$. We also create a clique $k_y$ of size $|C|+|Y|-1$.
    \item For each clause $c \in C$, we create a vertex $v_c \in V_C$.
    \item From each vertex $v_u \in V_U$, we draw an arc $(v_u, v_c)$ to every clause in which the positive literal $u$ appears. Similarly, we draw an arc $(v_{\bar{u}}, v_c)$ from each $v_{\bar{u}} \in V_{\bar{U}}$ to every clause in which the negative literal $\neg u$ appears.
    \item From every $v_x$, we draw an arc to one node in $k_x$, and do the same thing with $v_{\bar{x}}$ and $k_{\bar{x}}$. We also draw an undirected edge between each $v_x$ and $v_{\bar{x}}$.
    \item Finally, from each $v_y$ and each $v_{\bar{y}}$, we draw an arc to one node in $k_y$.
\end{itemize}
An example of this construction can be found in \hyperref[Fig3]{Figure \ref{Fig3}}. We set $\Omega = |X|$, $\Phi = |X| + |Y|$, $K = (|X| +|Y|) \times (|Y|+|C|) + |C|-1$ and argue that answering {\sc Vaccination-Attack}$_{dir}$ on this instance is the same as answering $B_2^{CNF}$.

Indeed, if we are a given a solution to a \textit{Yes} instance of $B_2^{CNF}$, then by vaccinating the vertices corresponding to the opposite of the 0-1 assignment of $X$, we oblige the attacker to infect the vertices corresponding to the truth values for $X$. From there, by following the same reasoning as before, it is easy to see that the \textit{Yes} instance of $B_2^{CNF}$ leads to a \textit{Yes} instance of {\sc Vaccination-Attack}$_{dir}$, \ie~ the attacker cannot infect more than $K$ vertices.

Conversely, we show that a set $D^*$ corresponding to a solution of a \textit{Yes} instance of {\sc Vaccination-Attack}$_{dir}$ is a solution to a \textit{Yes} instance of $B_2^{CNF}$. The first thing to notice is that given that the vaccination budget is $\Omega = |X|$, that the size of the cliques $k_x$ and $k_{\bar{x}}$ is equal to $|C|+|Y|-1$ and that each clique can be disconnected from the graph by spending only one unit of vaccination budget, we necessarily have that the best vaccination strategy $D^* \subset \underset{x \in X}{\cup}\{v_x, v_{\bar{x}} \}$. Next, we show that the defender would be worse off is she decides to vaccinate both $v_{x'}$ and $v_{\bar{x'}}$ for some $x'\in X$ instead of vaccinating exactly one of each member of $\{v_x, v_{\bar{x}}\}$. In the best case scenario, in addition  to the vertices already vaccinated, deciding to vaccinate the two members of a pair will allow her to protect $|C|-1$ nodes in $V_C$ \textit{(it is not possible to remove all the arcs between the $V_U \cup V_{\bar{U}}$ and the $V_C$ as we suppose that $Y \neq \emptyset$, thus at least one clause contains a variable from $Y$)}. But by doing so,  as $\Omega = |X|$, the defender will also not protect at all a group of vertices $\{v_{x''}, v_{\bar{x}''}\} \cup V_{k_{x''}} \cup V_{k_{\bar{x}''}}$. Thus, the attacker can then spend only one unit of her own budget to attack all of this group, a quantity of infected vertices that otherwise would have been obtained by spending two units of his budget $\Phi$. Thus, defending the two members of $\{v_{x'}, v_{\bar{x}'}\}$ spared one unit of budget for the attacker, which she can then use to attack one of the disconnected cliques of size $|C|+|Y| - 1 > |C| -1$. Thus, making such a move for the defender is strictly worse than not doing it and $D^*$ contains exactly one vertex from each $\{v_x, v_{\bar{x}}\}$. \\
After this stage, it is easy to see that the best move for the attacker is to attack all of the $D^* \backslash (V_x \cup V_{\bar{x}})$, and for the variables in $Y$, the situation reduces to the one we already discussed with \textsc{Attack$_{dir}$} \textit{(note that it is always more interesting for the attacker to spend her budget on attacking the $v_y$ and $v_{\bar{y}}$ than the disconnected cliques as it will always infect more vertices)}. Hence, in the end, if the attacker did not manage to infect strictly more than $(|Y| + |X|) \times (|Y|+|C|) + |C|-1$ vertices, it means that at least one clause is false, which concludes the proof.
\end{proof}
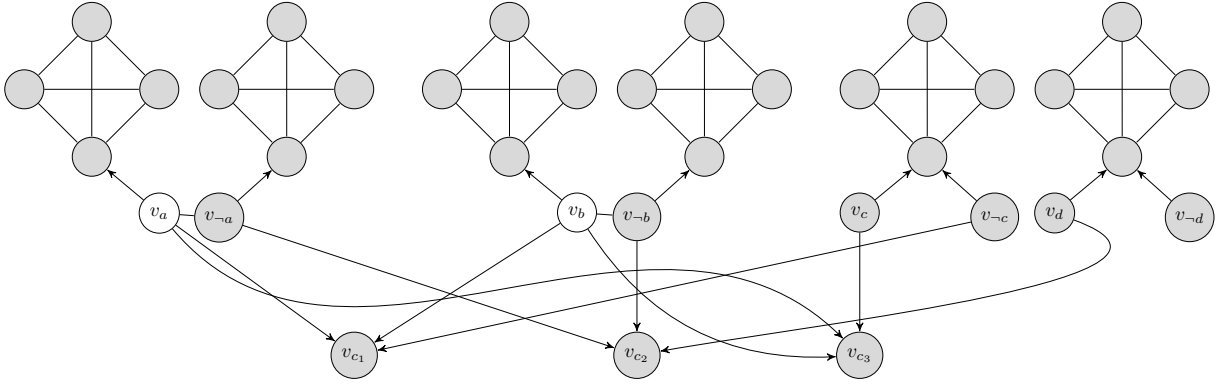
\begin{figure}[ht]
    \centering
    \resizebox{\columnwidth}{!}{
    \input{Paper_Version/Figure3}}
    \caption{Example of construction of $G$ from the boolean expression in CNF with 3 literals in each clause $E = (a \lor b \lor \neg c) \land (\neg a \lor \neg b \lor d) \land (a \lor c \lor b) $. Here, $X = \{a, b \}$ and $Y = \{c,d\}$. Taking $D = \{v_a, v_b\}$, \ie, obliging both $a$ and $b$ to be \textit{False} makes it impossible to satisfy $E$.}
    \label{Fig3}
\end{figure}
\begin{cor}
{\sc MCN}$_{dir}$ is $\Sigma_2^p$-hard.
\end{cor}

%% file: Paper_Version/Figure2.tex
\begin{tikzpicture}[->,>=stealth',shorten >=0.2pt,auto,node distance=1cm,baseline={(current bounding box.north)},  scale=0.8, transform shape,
  main node/.style={circle,fill=black!15,draw,minimum size=0.7cm},spe node/.style={fill=white,draw=white}, inf node/.style={fill=black,draw=white,circle,minimum size=0.7cm}]

  \node[main node] (1) {};
  \node[main node] (2) [right=of 1] {};
  \node[main node] (3) [below right =of 2] {};
  \node[main node] (4) [below left=of 1] {};
  \node[main node] (5) [below right= 0.5cm and 1.5cm of 4] {};
  \node[inf node] (6) [below left=of 5] {\textcolor{white}{$v_a$}};
  \node[main node] (7) [below right=of 5] {$v_{\neg a}$};
  \node[main node] (8) [below =3cm of 5] {$v_{c_1}$};
  \node[main node] (9) at (5.5,-0) {};
  \node[main node] (10) [right=of 9] {};
  \node[main node] (11) [below right =of 10] {};
  \node[main node] (12) [below left=of 9] {};
  \node[main node] (13) [below right= 0.5cm and 1.5cm of 12] {};
  \node[inf node] (14) [below left=of 13] {\textcolor{white}{$v_b$}};
  \node[main node] (15) [below right=of 13] {$v_{\neg b}$};
  \node[main node] (16) [below =3cm of 13] {$v_{c_2}$};
  \node[main node] (17) at (11,0) {};
  \node[main node] (18) [right=of 17] {};
  \node[main node] (19) [below right =of 18] {};
  \node[main node] (20) [below left=of 17] {};
  \node[main node] (21) [below right= 0.5cm and 1.5cm of 20] {};
  \node[inf node] (22) [below left=of 21] {\textcolor{white}{$v_c$}};
  \node[main node] (23) [below right=of 21] {$v_{\neg c}$};
  \node[main node] (24) [below =3cm of 21] {$v_{c_3}$};

  \path[every node/.style={font=\sffamily\small}]
    (1) edge node {}  (2)
    (2) edge node {}  (3)
    (4) edge node {}  (1)
    (5) edge node {}  (4)
    (6) edge node {}  (5)
        edge node {}  (8)
    (7) edge node {}  (5)
    (9) edge node {}  (10)
    (10) edge node {}  (11)
    (12) edge node {}  (9)
    (13) edge node {}  (12)
    (14) edge node {}  (13)
    (15) edge node {}  (13)
    (14) edge node {}  (16)
     (17) edge node {}  (18)
    (18) edge node {}  (19)
    (20) edge node {}  (17)
    (21) edge node {}  (20)
    (22) edge node {}  (21)
    (23) edge node {}  (21)
    (14) edge node {} (8)
    (23)  edge node {} (8)
    (7) edge node {} (16)
    (22) edge node {} (16)
    (6) edge node {} (24)
    (15)  edge node {} (24)
    (22) edge node {} (24);
\end{tikzpicture}

%% file: Paper_Version/Figure3.tex
\begin{tikzpicture}[->,>=stealth',shorten >=0.2pt,auto,node distance=1cm,baseline={(current bounding box.north)},  scale=0.8, transform shape,
  main node/.style={circle,fill=black!15,draw,minimum size=0.7cm},spe node/.style={fill=white,draw=white,minimum size=0.7cm}, X node/.style={fill=white,draw=black,circle,minimum size=0.7cm}]

  \node[main node] (1) {};
  \node[main node] (2) [below left=of 1] {};
  \node[main node] (4) [below right =of 2] {};
  \node[main node] (3) [below right=of 1] {};
  \node[X node] (5) [below =1.5cm of 3] {$v_a$};
  \node[main node] (6) at (3.5,0) {};
  \node[main node] (7) [below left=of 6] {};
  \node[main node] (8) [below right =of 7] {};
  \node[main node] (9) [below right=of 6] {};
  \node[main node] (10) [below =1.5cm of 7] {$v_{\neg a}$};
  \node[main node] (11) at (7.5,0) {};
  \node[main node] (12) [below left=of 11] {};
  \node[main node] (14) [below right =of 12] {};
  \node[main node] (13) [below right=of 11] {};
  \node[X node] (15) [below =1.5cm of 13] {$v_b$};
  \node[main node] (16) at (11,0) {};
  \node[main node] (17) [below left=of 16] {};
  \node[main node] (19) [below right =of 17] {};
  \node[main node] (18) [below right=of 16] {};
  \node[main node] (20) [below =1.5cm of 17] {$v_{\neg b}$};
  \node[main node] (21) at (15,0) {};
  \node[main node] (22) [below left=of 21] {};
  \node[main node] (24) [below right =of 22] {};
  \node[main node] (23) [below right=of 21] {};
  \node[main node] (25) [below =1.5cm of 22] {$v_c$};
  \node[main node] (31) [below =1.5cm of 23] {$v_{\neg c}$};
  \node[main node] (26) at (18.5,0) {};
  \node[main node] (27) [below left=of 26] {};
  \node[main node] (29) [below right =of 27] {};
  \node[main node] (28) [below right=of 26] {};
  \node[main node] (30) [below =1.5cm of 27] {$v_d$};
  \node[main node] (32) [below =1.5cm of 28] {$v_{\neg d}$};
  \node[main node] (33) [below =4cm of 9] {$v_{c_1}$};
  \node[main node] (34) [below =4cm of 17] {$v_{c_2}$};
  \node[main node] (35) [below =4cm of 22] {$v_{c_3}$};
  
  \path[every node/.style={font=\sffamily\small}]
    (1) edge[-] node {}  (2)
        edge[-] node {}  (3)
        edge[-] node {}  (4)
    (2) edge[-] node {}  (3)
        edge[-] node {}  (4)
    (3) edge[-] node {}  (4)
    (5) edge node {}  (4)
    (6) edge[-] node {}  (7)
        edge[-] node {}  (8)
        edge[-] node {}  (9)
    (7) edge[-] node {}  (8)
        edge[-] node {}  (9)
    (8) edge[-] node {}  (9)
    (10) edge node {}  (8)
    (5) edge[-] node {} (10)
    (11) edge[-] node {}  (12)
        edge[-] node {}  (13)
        edge[-] node {}  (14)
    (12) edge[-] node {}  (13)
        edge[-] node {}  (14)
    (13) edge[-] node {}  (14)
    (15) edge node {}  (14)
    (16) edge[-] node {}  (17)
        edge[-] node {}  (18)
        edge[-] node {}  (19)
    (17) edge[-] node {}  (18)
        edge[-] node {}  (19)
    (18) edge[-] node {}  (19)
    (20) edge node {}  (19)
    (15) edge[-] node {} (20)
    (25) edge node {} (24)
    (21) edge[-] node {}  (22)
        edge[-] node {}  (23)
        edge[-] node {}  (24)
    (22) edge[-] node {}  (23)
        edge[-] node {}  (24)
    (23) edge[-] node {}  (24)
    (31) edge node {}  (24)
    (26) edge[-] node {}  (27)
        edge[-] node {}  (28)
        edge[-] node {}  (29)
    (27) edge[-] node {}  (28)
        edge[-] node {}  (29)
    (28) edge[-] node {}  (29)
    (30) edge node {}  (29)
    (32) edge node {} (29)
    (5) edge node {} (33)
        edge[out=-50] node {} (35)
    (10) edge node {} (34)
    (15) edge node {} (33)
         edge[bend right] node {} (35)
    (20) edge node {} (34)
    (25) edge node {} (35)
    (31) edge node {} (33)
    (30) edge[out=-20,in=10] node {} (34);
\end{tikzpicture}

%% file: Paper_Version/Special_Cases.tex
\section{{\sc Protection}: tractability limits}\label{sec:special}

In this section, we will concentrate on optimal protection strategies given $I$ (directly infected vertices). Without loss of generality, in what follows, we  restrict our attention to the induced graph obtained by considering only  non-saved  vertices when there is no protection (the remaining are already saved, even without no protected vertices).

The motivation to provide a closer look to the protection problem in the unitary cases (undirected and directed graphs) is based on the fact that their NP-hardness was established for split graphs, while for the weighted case it was proven even for trivial graphs. Such results do not clarify the problem complexity for trees, or even graphs of bounded treewidth, neither for directed acyclic graphs (DAGs), polytrees and arborescences. Frequently, NP-complete problems on graphs become polynomially solvable on such graph classes. In section~\ref{subsec:protect_trees}, we describe a dynamic programming approach for trees to determine the optimal protection solution in polynomial time. We also connect our problem complexity with the results in monadic second-order logic for tree-decomposable graphs~\cite{langer2014,Barnetson2020}. In section~\ref{subsec:protect_dag}, we describe the problem properties for DAGs, making it simple to show that {\sc Protection}$_{dir}$ is NP-complete. We terminate this section by showing that the optimal protection strategy can be determined in polynomial time for arborescences. 


\subsection{{\sc Protect} over trees}\label{subsec:protect_trees}

We next focus on {\sc Protect}; recall that it is the case of undirected graphs with unitary costs and benefits. Results for a special version of this problem where only one vertex is infected, aka the {\sc Firefighter Problem}, already exist in the literature. The recent work of \cite{Barnetson2020} establishes that for this special version of {\sc Protect}, the decision version of the problem can be solved in linear time over graphs of bounded treewidth, through the use of a reformulation in \emph{Extended Monadic Second Order} (EMSO) logic. We first extend this result to the case of an arbitrary number of infected vertices to show that {\sc Protect} is solvable in polynomial time over graphs of bounded treewidth. 

\begin{lem}\label{prop_ProtectTreewidth}
{\sc Protect} can be solved in polynomial time over graphs with constant bounded treewidth.
\end{lem}
\begin{proof}
The key factor is to reformulate our problem in terms of an MSO-formula $\varphi$ based on set variables, which captures the graph structure of the problem, and an evaluation relation $\psi$ over a set of integer variables, which captures the ``number'' aspect of the problem. In order to do so, we will define as in \cite{Barnetson2020} two sets $P$ and $X$ where $P$ is the set of protected vertices that separate the infected vertices of set $I$ from the saved vertices which are not protected $X$. Apart from the classic universal quantifier and logical connectives, we need to make use of a binary relation $adj(x,y)$ to assess the adjacency of two vertices $x$ and $y\in V$: $adj(x,y)$ true if $(x,y) \in E$, and false otherwise. The definitions of $\varphi$ and $\psi$ for {\sc Protect} are the following:
\begin{eqnarray*}
\varphi &=& (\forall v(v\in I\Rightarrow(v\notin P)\land(v\notin X)))\land(\forall x(x\in P)\Rightarrow(x\notin X))\land\\
&& (\forall x\forall y((x\in X)\land(adj(x,y))\land(y\notin P))\Rightarrow(y\in X)).
\end{eqnarray*}
\[ \psi = ( \vert P \vert \leq \Lambda)\land( \vert X\vert +\vert P \vert \geq K), \]
where the aim is to save at least $K$ vertices. 
Through the above expression for $\varphi$, it is established that the sets $P$ and $X$ have an empty intersection and that any neighbour of a vertex in $X$ must be either  in $P$ or $X$ itself, \ie, the set $P$ is a separator for the sets $I$ and $X$. Since the above definitions respect the limitations of MSO logic formulations, a theorem due to \cite{Arnborg1991} implies that the problem can be solved in $O(f(w)\cdot \vert V \vert)$ where $f$ is a function of the treewidth $w$ of the graph. Consequently, we can conclude that {\sc Protect} can be solved in polynomial time for graphs whose treewidth is bounded by a constant.
\end{proof}

Even though the above theorem is powerful from a theoretical perspective, it is of little practical use, as underlined in \cite{langer2014}. Indeed, the function $f(w)$ in the worst case complexity formula grows extremely fast with $w$ and the algorithm suffers from space problems in practical implementations. Therefore, \emph{Dynamic Programming} (DP) is often used to provide more efficient algorithms. In this section we propose a DP algorithm to solve the optimization problem associated with  {\sc Protect} on trees. 
We consider a recursion scheme that works with growing subtrees, starting from the leaves and climbing up to the root vertex, solving the optimization problem on each subtree recursively and merging them as needed at each step of the recursion. This recursion scheme bears similarities with the scheme proposed in \cite{DisGroLoc11cnp} for the pairwise CNDP over trees. In the following, we label a vertex as \emph{attacked} when it has been directly infected by the attacker while the term \emph{infected} is used both for directly attacked vertices and indirectly infected ones, and \emph{secondary infected} vertices is used only for vertices indirectly infected after the initial attack.

For further analysis, we denote by $\T_a$ the subtree of tree $\T$ rooted at vertex $a\in V$, 
and by $a_i$ with $i\in\{1,...,s\}$ the children of $a$.
We define as $\T_{a_{i\rightarrow s}}$ the subtree constituted
by $\{a\}\cup_{j=i,...,s}\T_{a_j}$. An example of a tree $\T$ rooted at vertex $a$ is depicted in
Figure~\ref{Fig:subtrees} where subtree $\T_{a_2}$ is represented by diamond shaped vertices while
subtree $\T_{a_{3\rightarrow 4}}$ is represented by round shaped vertices.
All recursions in our dynamic programming approaches are based on traversing the tree in postorder
({\ie} from the leaves to the root) and from the right part of each tree level to the left one. For example in Figure~\ref{Fig:subtrees}, once the recursive functions are computed and saved for subtrees $\T_{a_{3}}$ and $\T_{a_4}$, we will compute the recursion functions associated to $\T_{a_{3\rightarrow 4}}$ by merging the results for both subtrees in both situations when $a$ is vaccinated, infected, protected or neither of these possibilities. We consider that tree $\T$ is rooted at vertex $r$.

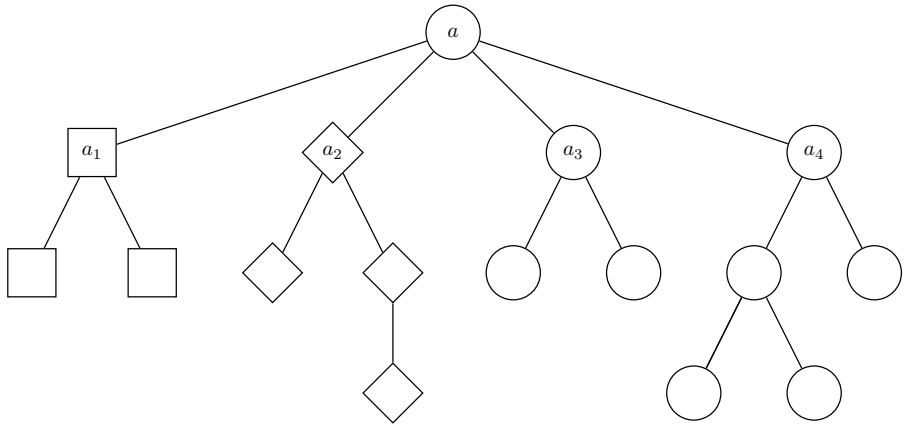
\begin{figure}[ht]
\centering
\scalebox{0.8}{
\begin{tikzpicture}[>=latex, semithick]

\node [circle, draw, semithick,  minimum size=0.9cm] (1) at (7,6) {$a$};
\node [draw, semithick,  minimum size=0.8cm] (2) at (1,4) {$a_1$};
\node [diamond, draw, semithick,  minimum size=1cm] (3) at (5,4) {$a_2$};
\node [circle, draw, semithick,  minimum size=0.9cm] (4) at (9,4) {$a_3$};
\node [circle, draw, semithick,  minimum size=0.9cm] (5) at (13,4) {$a_4$};
\node [draw, semithick,  minimum size=0.8cm] (6) at (0,2) {};
\node [draw, semithick,  minimum size=0.8cm] (7) at (2,2) {};
\node [diamond, draw, semithick,  minimum size=1cm] (8) at (4,2) {};
\node [diamond, draw, semithick,  minimum size=1cm] (9) at (6,2) {};
\node [diamond, draw, semithick,  minimum size=1cm] (10) at (6,0) {};
\node [circle, draw, semithick,  minimum size=0.9cm] (11) at (8,2) {};
\node [circle, draw, semithick,  minimum size=0.9cm] (12) at (10,2) {};
\node [circle, draw, semithick,  minimum size=0.9cm] (13) at (12,2) {};
\node [circle, draw, semithick,  minimum size=0.9cm] (14) at (14,2) {};
\node [circle, draw, semithick,  minimum size=0.9cm] (15) at (11,0) {};
\node [circle, draw, semithick,  minimum size=0.9cm] (16) at (13,0) {};

\draw [-, semithick](1) -- (2);
\draw [-, semithick](1) -- (3);
\draw [-, semithick](1) -- (4);
\draw [-, semithick](1) -- (5);
\draw [-, semithick](2) -- (6);
\draw [-, semithick](2) -- (7);
\draw [-, semithick](3) -- (8);
\draw [-, semithick](3) -- (9);
\draw [-, semithick](9) -- (10);
\draw [-, semithick](4) -- (11);
\draw [-, semithick](4) -- (12);
\draw [-, semithick](5) -- (13);
\draw [-, semithick](5) -- (14);
\draw [-, semithick](13) -- (15);
\draw [-, semithick](13) -- (15);
\draw [-, semithick](13) -- (16);

\end{tikzpicture}}
\caption{Example of a tree with subtree $\T_{a_2}$ represented by diamond shaped vertices and
subtree $\T_{a_{3\rightarrow 4}}$ represented by round shaped vertices.}
\label{Fig:subtrees} 
\end{figure}

We introduce the following recursion functions:

\medskip
\begin{parcolumns}[colwidths={1=.23\textwidth},distance=0em]{2}
\colchunk{%
 $F_a(c,m,\sigma)$ :=
}
\colchunk{%
maximum number of saved vertices in subtree $\T_a$ when $c$ vertices have been protected, $m$ unprotected vertices in $\T_a$ are linked to $a$ by an unprotected but non-infected path (including $a$ itself) and $\sigma=1$ if an attacked vertex in $\T_a$ is linked to $a$ by an unprotected path (including $a$) and $\sigma=0$ otherwise.
}
\end{parcolumns}
\begin{parcolumns}[colwidths={1=.23\textwidth},distance=0em]{2}
\colchunk{%
 $G_{a_i}(c,m,\sigma)$ :=
}
\colchunk{%
maximum number of saved vertices in subtree $\T_{a_{i\rightarrow s}}$ when $c$ vertices have been protected, $m$ unprotected vertices in $\T_{a_{i\rightarrow s}}$ are linked to $a$ by an unprotected but non-infected path (including $a$ itself) and $\sigma=1$ if an attacked vertex in $\T_{a_{i\rightarrow s}}$ is linked to $a$ by an unprotected path (including $a$) and $\sigma=0$ otherwise.
}
\end{parcolumns}
Using the previously described functions, we can define the following recursions. The initial conditions for each leaf vertex $a$ and rightmost subtree $T_{a_s}$ are as follows:
\begin{equation}
F_a(c,m,\sigma) =
\begin{cases}
0      & \text{if $(c=0\land m=0\land \sigma=1)$} \cr
1      & \text{if $(c=1\land m=0\land \sigma=0)\lor (c=0\land m=1\land \sigma=0)$} \cr
-\infty & \text{otherwise ({\ie}, infeasible configurations)} \cr
\end{cases};
\end{equation}
\begin{equation}\label{eq:Gas}
G_{a_s}(c,m,\sigma) =
\begin{cases}
\max\left\lbrace F_{a_s}(c-1,m^\prime,\sigma^\prime)+m^\prime(1-\sigma^\prime):\ m^\prime=0,\dots,|\T_{a_s}|;\right.\cr \left.\sigma^\prime=0,1\right\rbrace \cr\qquad\text{if $a$ is protected ($m=\sigma=0$)} \cr
\max\left\lbrace F_{a_s}(c,m^\prime,\sigma^\prime)+m^\prime(1-\sigma^\prime)(1-\sigma):\ m^\prime=0,\dots,|\T_{a_s}|;\right.\cr \left.\sigma^\prime=0,1\right\rbrace \cr\qquad\text{if $a$ is vaccinated ($m=\sigma=0$) or attacked ($m=0$, $\sigma=1$)} \cr
F_{a_s}(c,m-1,\sigma)\cr\qquad \text{if $a$ is neither vaccinated, protected or attacked}. \cr
\end{cases}
\end{equation}
In Eq.~\eqref{eq:Gas}, the first case deals with a protected $a$ vertex and all $m^\prime$ unprotected vertices below $a_s$ are saved if $a_s$ is not linked to an attacked vertex inside $T_{a_s}$ ($\sigma^\prime=0$). The second case deals with either a vaccinated or attacked $a$ vertex so that the budget $c$ needs not be updated going from $T_{a_s}$ to $T_{a_{s}}\cup\{a\}$. The last case deals with an unattacked and unprotected $a$ vertex and parameter $m$ is incremented as the subtree is enlarged by vertex $a$.

The following equations handle the general case, for vertices which are neither leaf vertices or the root of rightmost subtrees:
\begin{equation}
F_a(c,m,\sigma) = G_{a_1}(c,m,\sigma) \quad \text{for a non-leaf vertex $a\in V$.}   
\label{eq:rec1}
\end{equation}
For each non-leaf vertex $a\in V$ and $i<s$: if $a$ is attacked
\begin{dgroup}
\begin{dmath}
G_{a_i}(c,0,1) = \max\left\lbrace
	F_{a_i}(c^\prime,m^\prime,\sigma^\prime) + G_{a_{i+1}}(c-c^{\prime},0,1) 
        \hiderel{:}\\
	\phantom{=\,} c^\prime \hiderel{=}0,\dots,|\T_{a_i}|;\ m^\prime \hiderel{=}0,\dots,|\T_{a_i}|;\ \sigma^\prime\hiderel{=}0,1\right\rbrace,
\label{eq:rec2}
\end{dmath}
\begin{dsuspend}
if $a$ is protected (either from vaccination or protection)
\end{dsuspend}
\begin{dmath}
G_{a_i}(c,0,0) = \max\left\lbrace
	F_{a_i}(c^\prime,m^\prime,\sigma^\prime) + G_{a_{i+1}}(c-c^\prime,0,0) + m^\prime(1-\sigma^\prime)
        \hiderel{:}\\
	\phantom{=\,} c^\prime \hiderel{=}0,\dots,|\T_{a_i}|; m^\prime \hiderel{=}0,\dots,|\T_{a_i}|;\ 
	\sigma^\prime\hiderel{=}0,1\right\rbrace,
	\label{eq:rec3}
\end{dmath}
\begin{dsuspend}
otherwise, if $a$ is neither protected nor infected
\end{dsuspend}
\begin{dmath}
G_{a_i}(c,m,\sigma) = \max\left\lbrace
	F_{a_i}(c^\prime,m^\prime,\sigma^\prime) + G_{a_{i+1}}(c-c^\prime,m-m^\prime,\sigma^{\prime\prime}) + m^\prime(1-\sigma)\delta_{ar}
        \hiderel{:}\\
	\phantom{=\,} c^\prime \hiderel{=}0,\dots,|\T_{a_i}|; m^\prime \hiderel{=}0,\dots,|\T_{a_i}|;\\ 
	\phantom{=\,} \sigma^\prime,\sigma^{\prime\prime}\hiderel{=}0,1\hiderel{:}\,\sigma\hiderel{=}\max\{\sigma^\prime,\sigma^{\prime\prime}\}\right\rbrace.
	\label{eq:rec4}
\end{dmath}
\end{dgroup}

Equation~\eqref{eq:rec2} focuses on the case where vertex $a$ is infected. In this case, no additional vertex is saved as all vertices below $a$ which were not infected in $\T_a$ and with an unprotected path to $a$ will be infected themselves, therefore the total number of saved vertices is the sum of already saved vertices from the two merged subtrees. Equation~\eqref{eq:rec3} regards the case of a protected $a$ vertex, either through earlier vaccination or through protection. In this case, the vertices under $a_i$ who were unprotected, linked to $a_i$ by an unprotected path and who are not in contact with an infected vertex through an unprotected path are confirmed saved and added to the cost function, additionally to the already saved vertices from both subtrees. Finally, Equation~\eqref{eq:rec4} deals with the last case where $a$ is neither infected nor protected in any way. In this case, the cost of the objective function is updated by the number of unprotected vertices linked to $a_i$ by an unprotected path, but only in the case that $a$ is the root vertex $r$ ($\delta_{ar}=1$ if $a=r$ and 0 otherwise) and $a$ is not linked to an infected vertex through an unprotected path ($\sigma=0$). Otherwise, we cannot ensure that the unprotected vertices below $a$ will be saved in the optimal solution.

The optimal value for the problem is given by the quantity 
$$\max\left\{F_r(c,m,\sigma):\,c=0,\dots,\Lambda;\,m=0,\dots,n;\,\sigma=0,1\right\}$$ where $r$ is the root vertex of the tree, since it represents the maximum number of saved vertices for each protection budget and the solution can be recovered by backtracking.
Considering the proposed dynamic program, we can state the following proposition.

\begin{thm}\label{prop_ProtectTrees}
{\sc Protect} over trees admits a polynomial time algorithm with time complexity $O(n^5)$.
\end{thm}
\begin{proof}
The number of functions $F_a(\cdot)$ and $G_a(\cdot)$ to compute for each value of $c$, $m$ and $\sigma$ is bounded by $2n^2$. The recursion steps involved in Equation~\eqref{eq:rec4} are bounded by $2n^2$ operations at most. Considering all vertices $n$, the running time of the dynamic programming algorithm is thus bounded by $O(n^5)$.
\end{proof}

Since the lower level of the problem over trees is polynomial, the {\MCN} over trees cannot be $\Sigma_3^p$-hard. Following a classic trick for DP algorithm, see {\eg} \cite{DisGroLoc11cnp}, a similar algorithm can be devised when vertices have protection costs and unit benefits, which remains polynomial. 
When both types of weights are integer, the algorithms become pseudo-polynomial and the problem becomes weakly NP-hard.

\subsection{{\sc Protection}$_{dir}$ over directed acyclic graphs}\label{subsec:protect_dag}

We will show that an optimal protection strategy can be restricted to \emph{{\cand}} vertices.

\begin{definition}
In a directed graph $G=(V,A)$, a vertex $v \in V\setminus I$ that can be reached from a vertex of $I$ by a directed path and whose isolated protection results in a maximal set of saved vertices,  is called {\cand}.  Denote by $\mathcal{C}$ the set of {\cand} vertices.
\label{def:cand}
\end{definition}

In other words, a {\cand} vertex $v$ has no predecessor whose protection implies saving $v$. See Figure~\ref{fig:3graphs} for an illustration on popular graph sub-classes of DAGs. In the case of Figure~\ref{fig:cand}, $\mathcal{C}=\lbrace 1,2,3,9 \rbrace$; \eg, vertex 5 is not a {\cand}, since its protection saves vertices $\lbrace 6,7,8 \rbrace$, but this is also guaranteed by saving vertex 2 instead, resulting in the maximal set of saved vertices $\lbrace 2, 3,4,5,6,7,8\rbrace$. In Figure~\ref{fig:dag_b}, protecting vertex 1 suffices to save all the remaining non attacked vertices. Finally, in Figure~\ref{fig:arborescence}, the successors of the attacked vertices are exactly the set of candidates.

\begin{figure}[!ht]
\centering 
\begin{subfigure}{.3\textwidth}
\centering
\begin{tikzpicture}[->,>=stealth',shorten >=0.2pt,auto,node distance=1.5cm,baseline={(current bounding box.north)},  scale=0.75, transform shape,
  main node/.style={circle,fill=black!15,draw,minimum size=0.7cm}, inf node/.style={circle,fill=black,draw,minimum size=0.7cm},cand node/.style={circle,fill=black!15,draw, dashed,minimum size=0.7cm}]

  \node[inf node] (12) {12};
  \node[cand node] (1) [below of=12] {1};
  \node[inf node] (10) [right of=12] {10};
  \node[cand node] (3) [below of=10] {3};
  \node[cand node] (2) [below of=3] {2};
   \node[main node] (4) [below of=2] {4};
   \node[main node] (5) [below of=4] {5};
   \node[main node] (6) [below of=5] {6};
   \node[main node] (7) [below left= of 6] {7};
         \node[main node] (8) [below right = of 6] {8};
   \node[cand node] (9) [below of=8] {9};
      \node[inf node] (11) [left= of 9] {11};

  \path[every node/.style={font=\sffamily\small}]
    (12) edge node {}  (1)
    (10) edge node {}  (3)
    (1) edge node {}  (2)
    (3) edge node {}  (2)
    (2) edge node {}  (4)
    (4) edge node {}  (5)
    (5) edge node {}  (6)
    (6) edge node {}  (7)
    (6) edge node {}  (8)
    (8) edge node {}  (9)
    (11) edge node {}  (9);
\end{tikzpicture}
\caption{Graph induced by $V\setminus I$ is a polytree.}
\label{fig:cand}
\end{subfigure}%
\begin{subfigure}{.3\textwidth}
\centering
\begin{tikzpicture}[->,>=stealth',shorten >=0.2pt,auto,node distance=1.5cm,baseline={(current bounding box.north)},  scale=0.75, transform shape,
  main node/.style={circle,fill=black!15,draw,minimum size=0.7cm}, inf node/.style={circle,fill=black,draw,minimum size=0.7cm},cand node/.style={circle,fill=black!15,draw, dashed,minimum size=0.7cm}]

  \node[inf node] (0) {0};
  \node[cand node] (1) [below of=0] {1};
  \node[main node] (2) [below left= of 1] {2};
  \node[main node] (4) [below right= of 1] {4};
  \node[main node] (3) [below right= of 2] {3};
   \node[main node] (5) [below of=3] {5};
   
  \path[every node/.style={font=\sffamily\small}]
    (0) edge node {}  (1)
    (1) edge node {}  (2)
    (1) edge node {}  (4)
    (3) edge node {}  (5)
    (2) edge node {}  (3)
    (4) edge node {}  (3);
\end{tikzpicture}
\caption{Graph induced by\\ $V\setminus I$ is a DAG.}
\label{fig:dag_b}
\end{subfigure}%
\begin{subfigure}{.4\textwidth}
\centering
\begin{tikzpicture}[->,>=stealth',shorten >=0.2pt,auto,node distance=1.5cm,baseline={(current bounding box.north)},  scale=0.75, transform shape,
  main node/.style={circle,fill=black!15,draw,minimum size=0.7cm}, inf node/.style={circle,fill=black,draw,minimum size=0.7cm},cand node/.style={circle,fill=black!15,draw, dashed,minimum size=0.7cm}]

  \node[inf node] (0) {0};
  \node[cand node] (1) [below of=0] {1};
  \node[main node] (2) [below left= of 1] {2};
  \node[main node] (3) [below left= of 2] {3};
  \node[cand node] (4) [below right= of 2] {4};
   \node[main node] (5) [below of=4] {5};
   \node[main node] (6) [below of=1] {6};
   \node[main node] (7) [below right= of 1] {7};
   \node[main node] (8) [below of=7] {8};
   \node[inf node] (9) [right of=3] {9};
   
  \path[every node/.style={font=\sffamily\small}]
    (0) edge node {}  (1)
    (1) edge node {}  (2)
    (2) edge node {}  (3)
    (2) edge node {}  (4)
    (4) edge node {}  (5)
    (9) edge node {}  (4)
    (1) edge node {}  (6)
    (1) edge node {}  (7)
    (7) edge node {}  (8);
\end{tikzpicture}
\caption{Graph induced by $V\setminus I$ is an arborescence.}
\label{fig:arborescence}
\end{subfigure}%
\caption{The set $I$ is represented by black vertices and {\cand} vertices are dashed.}
\label{fig:3graphs}
\end{figure}
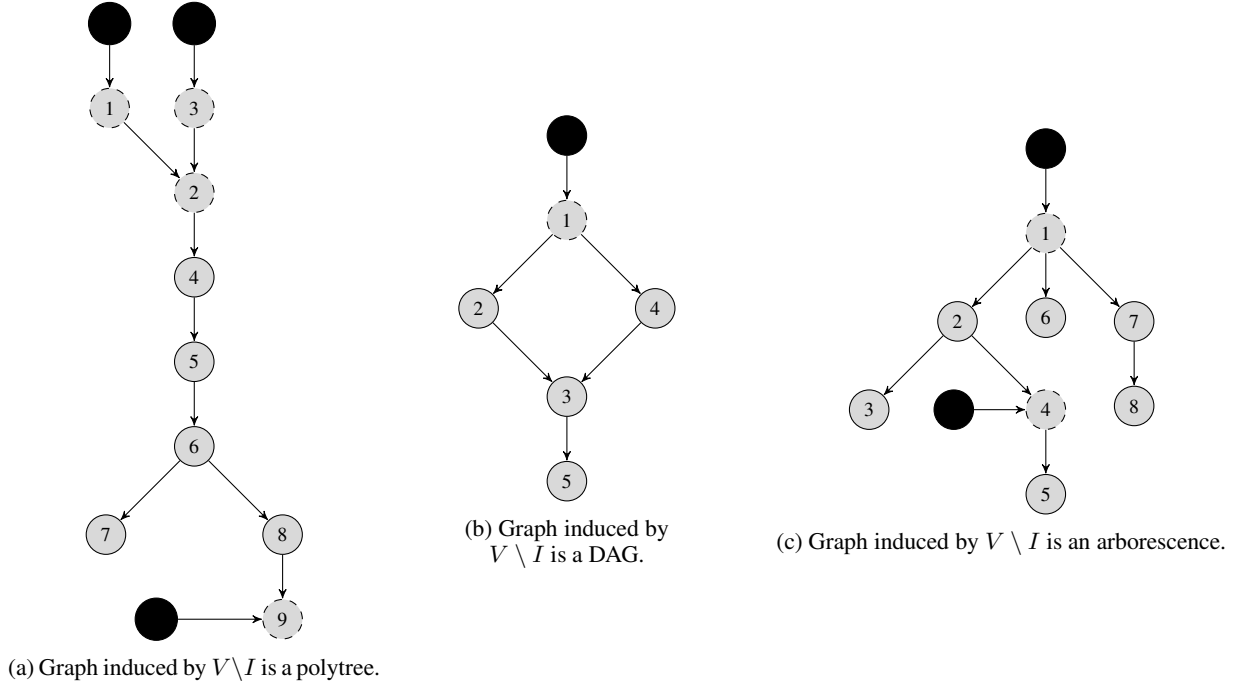

\begin{lem}
Let $G=(V,A)$ be a directed graph. Given $I$ and $\Lambda$, there is an optimal protection strategy $P \subseteq \mathcal{C}$.
\label{lem:cand}
\end{lem}
\begin{proof}
Let $P \subseteq V\setminus I$ be an optimal protection strategy such that exists $v \in P\setminus \mathcal{C}$. Then, by the definition of {\cand}, there is a vertex $u \in \mathcal{C}$ whose isolated protection implies saving $v$, as well as, all the vertices that $v$ alone was saving. Hence, a feasible protection strategy can be obtained by removing $v$ from $P$ and adding $u$ to $P$: note that either the used budget is maintained, if $u \notin P$, or decreased, if $u \in P$. Let this strategy be denoted by $\tilde{P}=\left(P-\lbrace v \rbrace \right) \cup \lbrace u \rbrace$.

By contradiction, suppose that $\tilde{P}$ is not optimal: there is some vertex $r$ that was saved in $P$ but not in $\tilde{P}$. In fact, we can conclude that under $P$, $r$ was saved due to $v$ being saved (protected) and possibly due to some other vertices in $P\setminus \lbrace v \rbrace \subseteq \tilde{P}$. However, under $\tilde{P}$, $v$ is also saved, as well, as the vertices in $P\setminus \lbrace v \rbrace $. Consequently, $r$ is saved in $\tilde{P}$, resulting in a contradiction.
\end{proof}

Furthermore, we can compute the \emph{value} of {\cand} vertices.

\begin{definition}
For each $v \in \mathcal{C}$, the value of $v$ is denoted by $p_v$ and it corresponds to the number of saved vertices if $v$ is the only protected vertex.
\label{def:value}
\end{definition}

In the example of Figure~\ref{fig:cand}, $p_1=1$, $p_2=6$, $p_3=1$ and $p_9=1$. However, note that this analysis does not make the problem trivial: in Figure~\ref{fig:cand}, if $\Lambda=2$, the optimal protection cannot be computed in a greedy way, \ie, protecting vertices 1 and 2 is not optimal; the only optimal solution is to protect vertices 1 and 3.

\begin{thm}
{\sc Protect}$_{dir}$ is NP-complete, even for directed acyclic graphs.
\label{th_protect_dir}
\end{thm}
\begin{proof}
The statement of {\sc Protect}$_{dir}$ is exactly the one of {\sc Protect} in Section~\ref{sec:undirected_uni}, except that the graph is directed. For sake of simplicity, we drop the subscript $a$ from $G_a$.

The problem is clearly in NP as given the protection $P$, the number of infected can be determined in polynomial time through a DFS.

Next, we reduce CNP$_{split}$ to  {\sc Protect}$_{dir}$, showing its NP-hardness. Given an instance of CNP$_{split}$, we build the following graph $G=(V,A)$:
\begin{itemize}
    \item For each $v \in \bar{V}_1$, we create the set of vertices $T_v = \lbrace t_v^1, t_v^2 \rbrace$ in $G$, and the arc $(t_v^1,t_v^2)$.
    \item For each $v \in \bar{V}_2$, we replicate it in $G$, and for each edge $(r,v) \in \bar{E}$ with $v \in \bar{V}_2$, the arc $(t_r^1,v)$ is added in $G$.
    \item Finally, we add the only attacked vertex $u$ to $G$ and connect it with each $t^1_v$ for $v \in \bar{V}_1$, through the arc $(u,t^1_v)$.
\end{itemize}
To complete the reduction it remains to set $\Lambda= B $ and $K=\lfloor 2+\sqrt{\vphantom{\prod^2}8\bar{K}+1} \rfloor$ (obtained by solving $\bar{K}=\binom{\frac{K-1}{2}}{2} $). See Figure~\ref{Fig_dag_pro} for an illustration of the reduction.

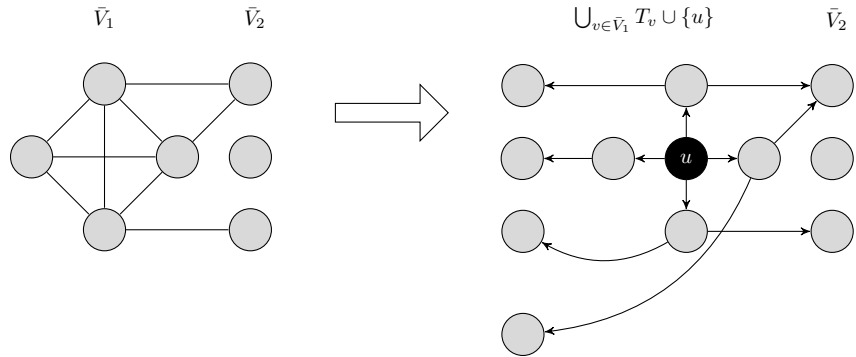
\begin{figure}[ht]
    \centering
    \input{Paper_Version/Figure_dag_protect}
    \caption{Example of construction of $G$ from $\bar{G}$.}
    \label{Fig_dag_pro}
\end{figure}

First, note that $\mathcal{C}$ of $G$ is $\lbrace t^1_v: v \in \bar{V}_1 \rbrace \cup \bar{V}_2$, where the vertices in the first set have value at least $2$, and the ones in the second have value 1. Hence, it is clear that the best protection strategy will prioritize the vertices $t^1_v$. In fact, we can argue than only those vertices can be in an optimal protection strategy. If $\Lambda=B \geq \bar{V}_1$, then the instance of CNP$_{split}$ is trivial. Therefore, we can assume $\Lambda=B < \bar{V}_1$ and thus, it holds  $P^* \subset \lbrace t^1_v: v \in \bar{V}_1 \rbrace.$ Consequently, choosing the optimal $P^*$ means to minimize the vertices in $T_v$, for $v \in \bar{V}_1$, and in $\bar{V}_2$ that are connected to $u$. By construction, those vertices connected with $u$ correspond to a connected component $\bar{C}_1$ in $\bar{G}$. Thus, $P^*$ minimizes the size of $\displaystyle \bigcup_{v \in \bar{C}_1} \lbrace t^2_v \rbrace \cup \lbrace u \rbrace \cup \bar{C}_1 $. The remaining of the proof follows an analogous reasoning to the proof of Theorem~\ref{th_protect_NP}.
\end{proof}

\subsubsection{Arborescence }

In this section we restrict the protection problem to the case where the graph induced by $V \setminus I$ is an arborescence.

\begin{definition}
A directed acyclic graph $G=(V,A)$ is an arborescence if its underlying undirected graph is a tree and there is a single vertex (root) that has a unique directed path from it to all other vertices.
\end{definition}

 In arborescence, the determination of $\mathcal{C}$ is straightforward. Since all vertices in $V\setminus I$ have  in-degree 1, either they are protected by their predecessor, and thus are not a {\cand}, or they are direct successors of vertices in $I$. Therefore, $\mathcal{C}$ is the set of all successors of vertices in $I$. For an illustration  see Figure~\ref{fig:arborescence}, where the vertices in $\mathcal{C}=\{1,4\}$ have in-degree 1. We can prove that in this case a greedy approach leads to optimality.

\begin{lem}
Given $G=(V,A)$, $I$ and $\Lambda$, if the graph induced by $V\setminus I$ is an arborescence, then an optimal protection can be determined in polynomial time, specifically, $O(\vert V \vert \log(|V|))$. Moreover, if the induced graph is a set of arborescences, the result also holds.
\label{lem:poly_arb}
\end{lem}
\begin{proof}
We start by showing that a greedy procedure runs in time $O(\vert V \vert \log(|V|))$.

As previously observed, for arborescences, the set of {\cand} vertices is easy to compute: it is the set of all successors of $I$. 

Next, the calculation of $p_v$ for each $v \in \mathcal{C}$ can be performed through a depth-first-search that records the saved vertices by {\cand s}. This requires $O(\vert V \vert)$ since the graph is an arborescence.

Finally, the $\Lambda$ {\cand} vertices of largest values are protected. This requires to order the vertices accordingly with $\lbrace p_v \rbrace_{v \in \mathcal{C}}$. Thus, the greedy method runs in $O(\vert V \vert \log(|V|))$.

Next, we show that the described method provides an optimal protection. Let $P$ be the obtained protection through the greedy method. The key idea to prove the optimality of $P$ is essentialy due to the fact that in an arborescence,  $\mathcal{C}$ is simply the set of all successors of $I$, otherwise, if we have a vertex of in-degree at least 2, we do not have an arborescence. Thus, the protection strategy $P$ cannot imply the protection of some {\cand} not in $P$. This shows the optimality of $P$.\end{proof}

Note that in trees (undirected graphs), it does not hold that $\mathcal{C}$ is the set of successors of the vertices in $I$. Hence, Lemma~\ref{lem:poly_arb} does not extend to the undirected case.

\begin{rem}
Note that in Lemmata~\ref{lem:cand} and ~\ref{lem:poly_arb}, we did not use the fact that $b_v=1$ $\forall v \in V$. Thus, it also holds when vertices' benefits are not unitary.
\end{rem}

%% file: Paper_Version/Figure_dag_protect.tex
\begin{tabular}{ccc}
\begin{tikzpicture}[-,>=stealth',shorten >=0.2pt,auto,node distance=1cm,baseline={(current bounding box.north)},  scale=0.8, transform shape,
  main node/.style={circle,fill=black!15,draw,minimum size=0.7cm},spe node/.style={fill=white,draw=white}]

  \node[main node] (1) {};
  \node[main node] (2) [below left=of 1] {};
  \node[main node] (4) [below right =of 2] {};
  \node[main node] (3) [below right=of 1] {};
  \node[main node] (5) [above right=of 3] {};
  \node[main node] (6)  [below=0.5cm of 5] {};
  \node[main node] (7) [below right=of 3] {};
  \node[spe node] (8) at (0,1.1) {$\bar{V}_1$};
  \node[spe node] (9) at (2.5,1.1) {$\bar{V}_2$};

  \path[every node/.style={font=\sffamily\small}]
    (1) edge node {}  (2)
        edge node {}  (3)
        edge node {}  (4)
        edge node {}  (5)
    (2) edge node {}  (3)
        edge node {}  (4)
    (3) edge node {}  (4)
        edge node {}  (5)
    (4) edge node {}  (7);
\end{tikzpicture}
& \hspace{1cm}
  \begin{tikzpicture}
    \useasboundingbox (-2,0);
    \node[single arrow,draw=black,fill=white,minimum height=1.5cm,shape border rotate=0] at (-2,-1.5) {};
  \end{tikzpicture}
  \hspace{1cm}
&
\begin{tikzpicture}[->,>=stealth',shorten >=0.2pt,auto,node distance=1cm,baseline={(current bounding box.north)},  scale=0.8, transform shape,
  main node/.style={circle,fill=black!15,draw,minimum size=0.7cm},spe node/.style={fill=white,draw=white}, inf node/.style={fill=black,draw,circle,minimum size=0.7cm}]

  \node[main node] (1) {};
  \node[main node] (11) [left=2cm of 1] {};
  \node[main node] (2) [below left=of 1] {};
  \node[main node] (12) [left=0.8cm of 2] {};
  \node[main node] (4) [below right =of 2] {};
  \node[main node] (13) [left=2cm of 4] {};
  \node[main node] (3) [below right=of 1] {};
  \node[main node] (14) [below=of 13] {};
  \node[main node] (5) [above right=of 3] {};
  \node[main node] (6)  [below=0.5cm of 5] {};
  \node[main node] (7) [below right=of 3] {};
  \node[spe node] (8) at (-0.7,1.1) {$\bigcup_{v \in \bar{V}_1} T_v \cup \lbrace u \rbrace$};
  \node[spe node] (9) at (2.5,1.1) {$\bar{V}_2$};
   \node[inf node] (10) [below=0.5cm of 1] {\textcolor{white}{$u$}};

  \path[every node/.style={font=\sffamily\small}]
    (1) edge node {} (5)
        edge node {} (11)
    (2) edge node {} (12)
    (3) edge[bend left] node {} (14)
        edge node {} (5)
    (4) edge node {}  (7)
        edge[bend left] node {} (13)
    (10) edge node {} (1)
         edge node {} (2)
         edge node {} (3)
         edge node {} (4);
\end{tikzpicture}
\end{tabular}